\documentclass[lettersize,journal]{IEEEtran}
\usepackage[dvipsnames]{xcolor}

\usepackage[utf8]{inputenc}
\usepackage[T1]{fontenc}

\usepackage{cite}
\usepackage{amsmath,amssymb,amsfonts}
\usepackage{amsthm,enumitem}
\usepackage{algorithmic}
\usepackage{graphicx}
\usepackage{textcomp}
\usepackage{xcolor}
\usepackage{color, soul}
\usepackage[]{footmisc}
\usepackage{bbm}
\usepackage{bm}
\usepackage{balance}

\usepackage{cuted}
\usepackage{stfloats}
\usepackage{mathtools,amssymb,lipsum, nccmath}
\newtheorem{theorem}{Theorem}
\newtheorem{lemma}{Lemma}
\newtheorem{corollary}{Corollary}
\newtheorem{prop}{Proposition}
\usepackage{algorithm,algorithmic}
\usepackage{color,soul}
\usepackage{graphicx}
\usepackage[caption=false,font=footnotesize]{subfig}
\usepackage{hyperref}
\usepackage{bm}
\newtheorem{Ass}{Assumption}

\newtheorem{remark}{Remark}

\newcommand{\mat}[1]{\mathbf{#1}}

\makeatletter
\newcommand{\vast}{\bBigg@{3}}
\makeatother

\begin{document}

\title{Distributed Optimization with Streaming Data:\\ A Temporal Weighting Perspective}
                   
\author{Muhammad Faraz Ul Abrar,~\IEEEmembership{Graduate Student Member, IEEE}, Nicol\`{o} Michelusi,~\IEEEmembership{Senior Member, IEEE}, \\and Erik G. Larsson,~\IEEEmembership{Fellow, IEEE}
\thanks{M. Faraz Ul Abrar and N. Michelusi are with the School of Electrical, Computer and Energy
Engineering, Arizona State University. email: \{mulabrar,
nicolo.michelusi\}@asu.edu. Erik G. Larsson is with the Department of Electrical Engineering (ISY),
Linköping University, 58183 Linköping, Sweden. e-mail: erik.g.larsson@liu.se. This research was funded in part by NSF under grant CNS-$2129015$. The work of E. G. Larsson was supported in part by ELLIIT, VR, and the KAW foundation.}
\thanks{A preliminary version of this work is available in \cite{DGD_TV_Asilomar26}.}
\vspace{-8mm}
}
\date{}
\maketitle

\setulcolor{red}
\setul{red}{2pt}
\setstcolor{red}

\begin{abstract}
Optimization theory is a widely used tool for intelligent decision-making. While classical optimization deals with fixed, time-invariant objective functions, many modern applications operate in dynamic environments where data arrive sequentially, and the learning objective evolves over time, often under decentralized data and communication constraints.
Motivated by these trends, we study decentralized optimization from streaming data through a structured time-varying formulation in which the global objective is a temporally weighted average of losses observed across the network.
We analyze multi-iteration decentralized first-order methods, including decentralized gradient descent 
and diffusion, where only a fixed number of communication/gradient steps can be performed before new samples arrive.  
For strongly convex and smooth losses, we develop guarantees for the Euclidean-norm \emph{tracking error} through a contraction-mapping viewpoint. The resulting bounds decompose the tracking error into a fixed-point tracking component and a bias term induced by decentralization and data heterogeneity.
We specialize our analysis to uniform and exponentially discounted weights, as well as their finite-memory \emph{windowed} counterparts.  The bounds explicitly characterize the roles of the temporal weighting rule, per-step iteration budget, step size, and network connectivity. Uniform weighting yields a vanishing fixed-point tracking contribution of order $\mathcal O(1/t)$, whereas discounted and windowed strategies generally induce non-vanishing tracking floors governed by the discount factor and effective memory, respectively. In all cases, decentralization induces an additional non-zero bias floor under a constant step size. Numerical experiments illustrate the predicted trends.
\end{abstract}


\vspace{-7mm}
\section{Introduction}
\label{sec:Intro}
Modern artificial intelligence (AI) systems increasingly rely on learning and optimization tools for intelligent decision-making across domains such as autonomous vehicles, robotics, telecommunications, power grids, and cyber-physical systems
\cite{Survey_dist_optm,Time_Structured}.
Despite their success, many existing AI systems follow an \emph{optimize-and-deploy} paradigm, where a model is trained for a fixed objective under an implicitly static data distribution.
In practice, however, many applications are both \emph{dynamic} and \emph{decentralized}: data arrive sequentially over time
\cite{DallAnese2019_stream,OL_survey,CL_Survey},
while storage, computation, and communication are distributed across multiple agents due to scale, privacy, or architectural constraints
\cite{Yuan_DFL_Survey2024,DL_wireless,NCOTA,biasedwirelessFL}.
Representative examples include mobile target localization, measurement-based network optimization, streaming data processing, and adaptive control under time-varying dynamics
\cite{Time_Structured,AliSayed2014Diffusion,DallAnese2019_stream,Bai2018}.
These settings call for learning systems that continuously adapt to incoming information while operating cooperatively over a network.

Time-varying optimization provides a natural framework for such problems. A network of agents collaboratively tracks the minimizer of an objective that evolves over time, typically while performing only a limited number of optimization and communication steps before the objective changes again
\cite{Time_Structured,DallAnese2019_stream}.
Exact tracking is therefore generally infeasible, and performance is naturally measured through the \emph{tracking error} (TE), defined at each time step as the Euclidean distance between the current iterate and the corresponding optimizer.
Existing decentralized time-varying optimization methods have primarily focused on algorithmic developments for improving TE performance
\cite{QLing2014ADMM,Usman2016DistributedDirected,SimonettoKoppel2016DecPC,Sun2017,RahiliRen2017,Bai2018,Bo2020,Huang2020DistributedTC,Sun2023}. However, these works treat the objective evolution as generic and largely \emph{unstructured}. In these formulations, the sequence of objectives is usually specified directly
as an arbitrary time-varying process and worst-case TE guarantees are derived under generic ``minimizer drift'' assumptions.
While broadly applicable, this viewpoint treats the temporal evolution as \emph{exogenous} and unstructured, and may therefore be conservative when the objective changes according to a specific mechanism.

In many modern learning problems, the objective does not evolve arbitrarily. Rather, it changes because new data samples are continuously acquired across the network, see \cite{Li2026NCCL} and references therein.
The objective evolution therefore inherits structure from the underlying data-acquisition process.
This observation motivates the central question addressed in this paper:
\emph{Can the streaming-data structure be explicitly incorporated into the optimization model to obtain sharper and more interpretable tracking guarantees?}

To address this question, in this work, we study decentralized time-varying optimization from a streaming-data perspective.
Each agent acquires data online, and the network objective at time $t$ is modeled as a \emph{temporally weighted} average of losses over the samples observed so far.
By explicitly encoding the temporal relevance of past data through the weights, the resulting formulation yields \emph{weight-specific} tracking guarantees that sharpen generic worst-case drift-based bounds.
Building on this structured formulation, we analyze decentralized first-order methods under limited computation and communication per time step.
In particular, we develop a unified TE analysis for a general class of temporal weighting strategies and specialize the results to two natural choices: (i) \emph{uniform} weights, which assign equal importance to all past samples and are natural for stationary environments, and (ii) \emph{exponentially discounted} weights, which geometrically discount old samples and emphasize recent observations.
We further specialize the guarantees to their finite-memory \emph{windowed} variants.
The resulting bounds explicitly characterize the roles of the temporal weighting rule, per-step iteration budget, network connectivity, step size, and data heterogeneity.
For smooth and strongly convex losses, our analysis reveals a vanishing fixed-point tracking contribution of order $\mathcal O(1/t)$ under uniform weighting and generally non-vanishing TE floors under discounted and finite-memory windowed strategies.
\vspace{0mm}
\subsection{Related Work}
The field of time-varying optimization traces back to early works on non-stationary optimization
\cite{Polyak_book,tracking_minimum_1998,Popkov},
primarily in the single-agent setting.
For smooth and strongly convex objectives, classical results establish tracking to a neighborhood whose size depends on a bounded-minimizer-drift or bounded-gradient-variation condition
\cite{Polyak_book,Popkov,Class_Prediction_Correction,Time_Structured}.
Methods that solely utilize the objective available at time $t$ to track the drifting minimizer are often termed \emph{correction-only} schemes \cite{Time_Structured}. In contrast, \emph{prediction-correction} schemes first predict the next optimizer using the objective at time $t$ (e.g., via a first-order optimality condition), and then perform a \emph{correction} update once the objective at time $t+1$ becomes available, leading to improved asymptotic tracking performance \cite{Class_Prediction_Correction,Fazlyab2018,Prediction_Correction_Constrained,SimonettoKoppel2016DecPC}. Related prediction mechanisms based on parameter estimation, including Kalman-filter and neural-network predictors, have also been studied \cite{Parameter_based_prediction,Simonetto2024Filter}. Time-varying optimization has also been investigated in decentralized multi-agent settings using primal methods, such as decentralized gradient descent (DGD) \cite{Nedic2009DGD}, and primal-dual methods, such as alternating direction method of multipliers (ADMM) \cite{Shi2014ADMM}; see, e.g., \cite{QLing2014ADMM,Usman2016DistributedDirected,SimonettoKoppel2016DecPC,Sun2017,RahiliRen2017,Bo2020,Huang2020DistributedTC,Sun2023,QLing2020}. Nonetheless, these works characterize temporal variation through generic worst-case drift measures and treat the objective evolution as exogenous. In contrast, the present formulation ties the evolution of the objective explicitly to streaming data through the temporal weights.

Another closely related research direction is \emph{online learning} \cite{Zinkevich,OL_OCO,Hazan2016IntroductionOL,OL_survey}. In online learning, a learner chooses the decision before observing the current loss, and performance is commonly measured through regret against a fixed or time-varying comparator \cite{jadbabaie15Dynamic,Besbes15NonStatSO,Shahin2018DistDynamic}.  
The paradigm of \emph{continual learning} similarly considers sequential data or ``task'' arrivals, with a primary focus on \emph{catastrophic forgetting} and methods for retaining previously learned knowledge \cite{Learn_without_forget,CL_Review,CL_Survey,CL_Survey_Defy_Forget,Catastrophic_McCloskey,Wu2026CLReview}. While these directions are related through their sequential-data viewpoint, they primarily address information availability, regret, or task adaptation. Our focus is instead on \emph{tracking-error} guarantees for decentralized optimization algorithms operating under communication and computation constraints.

Learning from data streams has also been studied in server-coordinated federated settings \cite{MitraOFL,marfoq23FLstream,ChungHu25StreamFL}. Among these, \cite{marfoq23FLstream} is closest in spirit to our work, as it considers memory-constrained learning from streaming data and assigns relative importance weights to samples.  
Yet a key distinction is that \cite{marfoq23FLstream} adopts a statistical learning viewpoint: it designs the weights to improve \emph{generalization error} under distributional assumptions. 
In contrast, we provide an optimization-theoretic, \emph{weight-specific} characterization of the TE of decentralized first-order methods, without imposing a data-generating distribution. Moreover, the setting considered in this work is fully decentralized and does not rely on a coordinating server.

\vspace{-2mm}
\subsection{Contributions and Organization}

Our main contributions are summarized as follows:
\begin{itemize}
[leftmargin=*]
\item We extend our single-agent formulation in \cite{TV_asilomar25} to
decentralized multi-agent networks and formulate decentralized learning from streaming data samples through temporally weighted objectives. We introduce a structured family of kernel-parameterized temporal weighting strategies, including uniform and exponentially discounted weights, together with their finite-memory windowed variants. 

\item Leveraging the contraction-mapping framework of \cite{MHT}, we develop a
unified TE analysis for multi-iteration DGD and diffusion under limited
communication and computation per time step.
\item We develop TE guarantees that decompose the overall error into a fixed-point tracking component and a bias term capturing the effect of data heterogeneity across the network. The resulting bounds explicitly characterize the dependence on the per-step iteration budget, network connectivity, step size, and temporal weighting rule.

\item We specialize the TE guarantees to four canonical weighting strategies
and use numerical experiments to illustrate the predicted effects of key system parameters and temporal-weighting choices.
\end{itemize}

The remainder of the paper is organized as follows. Section~\ref{sec:SystemModel} presents the system model and streaming-data objective. Section~\ref{sec:weights} introduces the temporal weighting strategies. Section~\ref{sec:alg_perf} describes the decentralized algorithms and performance metrics. Section~\ref{sec:analysis} develops the TE analysis. Section~\ref{sec:numerics} presents numerical results. Section~\ref{sec:conclusion} concludes the paper.
\vspace{-4mm}
\subsection{Notation}
Scalars are denoted by italic letters (e.g., $a$), vectors by
boldface lowercase letters (e.g., $\mathbf w$), and matrices by
boldface uppercase letters (e.g., $\mathbf A$). Bold sans-serif lowercase letters are reserved for
network-stacked vectors. The all-ones vector in $\mathbb R^N$ is denoted by $\bm 1_N$.
For a vector or matrix, $(\cdot)^\top$ denotes its transpose, and the
Euclidean norm is denoted by $\|\cdot\|$. For a stacked vector
$\bm{\mathsf w}
= \begin{bmatrix}
        \mathbf w_{1}^{\top} 
        \cdots 
        \mathbf w_{N}^{\top}
\end{bmatrix}^{\top}
\in\mathbb R^{Nd}$,
the notation $[\bm{\mathsf w}]_n\in\mathbb R^d$ denotes its $n$-th
block.
The identity matrix of dimension $q$ is denoted by
$\mat I_q\in\mathbb R^{q\times q}$, and $\otimes$ denotes the
Kronecker product. For a matrix $\mat A$, $\|\mat A\|_2$ denotes its
spectral norm and $\mat A^\dagger$ its Moore--Penrose pseudoinverse. The symbols $\mathcal R(\mat{A})$ and $\mathcal N(\mat{A})$ denote the range and nullspace of $\mat{A}$, respectively. If $\mat{A}$ is symmetric, $\mat{A}\succeq 0$ means that $\mat{A}$ is positive semidefinite, and $\lambda_{\min}(\mat{A})$ and $\lambda_{\max}(\mat{A})$ denote its minimum and maximum eigenvalues.
The indicator function of an event $\mathcal \omega$ is denoted by $\mathbbm{1}\{\mathcal \omega\}$. For a scalar $x\in\mathbb{R}$, we define its positive and negative parts as $x_+ \triangleq \max\{x,0\}$ and $x_- \triangleq \max\{-x,0\}$. 
For functions $f(\epsilon)$ and $g(\epsilon)$, we write
$f(\epsilon)=\mathcal O(g(\epsilon))$ if $f(\epsilon)$ is upper bounded by a constant multiple of $g(\epsilon)$ as $\epsilon\to0$, and
$f(\epsilon)=\Theta(g(\epsilon))$ 
if $f(\epsilon)$ is both upper and lower bounded by constant multiples of $g(\epsilon)$ as $\epsilon\to0$.
\vspace{-3mm}
\section{System Model and Problem Formulation}
\label{sec:SystemModel}
We consider a network of $N$ agents collaboratively solving a time-varying optimization problem induced by streaming data, as shown in Fig.~\ref{fig:system_model}.
At each time step $t\ge 1$, agent $n\in\{1,\ldots,N\}$ observes a new data sample and incurs the instantaneous loss $\ell_{n,t}(\mathbf{w})$ evaluated at the model parameter $\mathbf{w} \in \mathbb{R}^d$. 
The network seeks to \emph{track} the minimizer of a suitably constructed weighted average of the accumulated losses (discussed next), under a limited communication/computation budget per time index. 
The agents communicate over an undirected, connected, time-invariant graph. In particular, the communication among agents is described by a symmetric mixing matrix $\mat M\in\mathbb R^{N\times N}$. Specifically, $[\mat M]_{m,n}$ is the weight used by agent $m$ to combine information from agent $n$, and $[\mat M]_{m,n}=0$ whenever agents $m$ and $n$ do not directly communicate. We assume that $\mat{M}$ satisfies
\begin{align}
\mat{M}=\mat{M}^\top, \qquad \mat{M}\bm{1}_N=\bm{1}_N.
\label{eq:mixing_matrix_assump}
\end{align}
We further assume that the eigenvalues of $\mat M$ lie in $(-1,1]$ and are ordered as $1=\lambda_1>\lambda_2\ge \cdots \ge \lambda_N>-1$.\footnote{Sufficient conditions for this 
spectral
property, e.g., in terms of nonnegativity of $\mat M$ and non-bipartiteness of the underlying graph, are discussed in \cite{MHT}.}

We next define the streaming-data objective. The network-wide instantaneous loss at time $t$ is
\begin{align}
    F_t(\mathbf{w})
    \triangleq \frac{1}{N}\sum_{n=1}^N \ell_{n,t}(\mathbf{\mathbf{w}}),
    \label{eq:centralized_loss}
\end{align}
capturing the average loss computed across all the agents at time $t$. 
Building on the single-agent formulation in \cite{TV_asilomar25}, we define the time-$t$ objective as a \emph{temporally weighted average} of the past network-wide losses.
Let $\{a_i(t)\}_{i=1}^t$ represent a nonnegative temporal weighting strategy that satisfies 
\begin{align}
    a_i(t)\in[0,1], \; \forall i\leq t, \; \quad
    \sum_{i=1}^t a_i(t) = 1
    \label{eq:weight_simplex_constraint}
\end{align}
for all $t\geq1$.
Several strategies that satisfy \eqref{eq:weight_simplex_constraint} are discussed in the next section.  
This leads to the following time-varying optimization problem:
\begin{align}
    \overline{\mathbf{w}}_t^*
    \in \arg\min_{\mathbf{w}\in\mathbb{R}^d} \overline{F}_t(\mathbf{w})\;,
    \quad 
    \overline{F}_t(\mathbf{w})\triangleq \sum_{i=1}^t a_i(t)\,F_i(\mathbf{w}),
    \label{eq:main_tv_prob}
\end{align}
defined for all $t \geq 1$. We emphasize that the objective $\overline F_t(\cdot)$, and hence its minimizer $\overline{\mathbf{w}}_t^*$, evolves as new samples arrive across the network. Our goal is to characterize how well decentralized methods, such as DGD \cite{Nedic2009DGD} and diffusion \cite{AliSayed2014Diffusion}, track this moving minimizer when only a limited number of communication/computation iterations can be performed per time step. Before presenting the decentralized algorithms and the associated tracking guarantees, we first discuss the temporal weighting strategies used in \eqref{eq:main_tv_prob}.

\begin{figure}
\centering
    \includegraphics[width=0.90\linewidth]{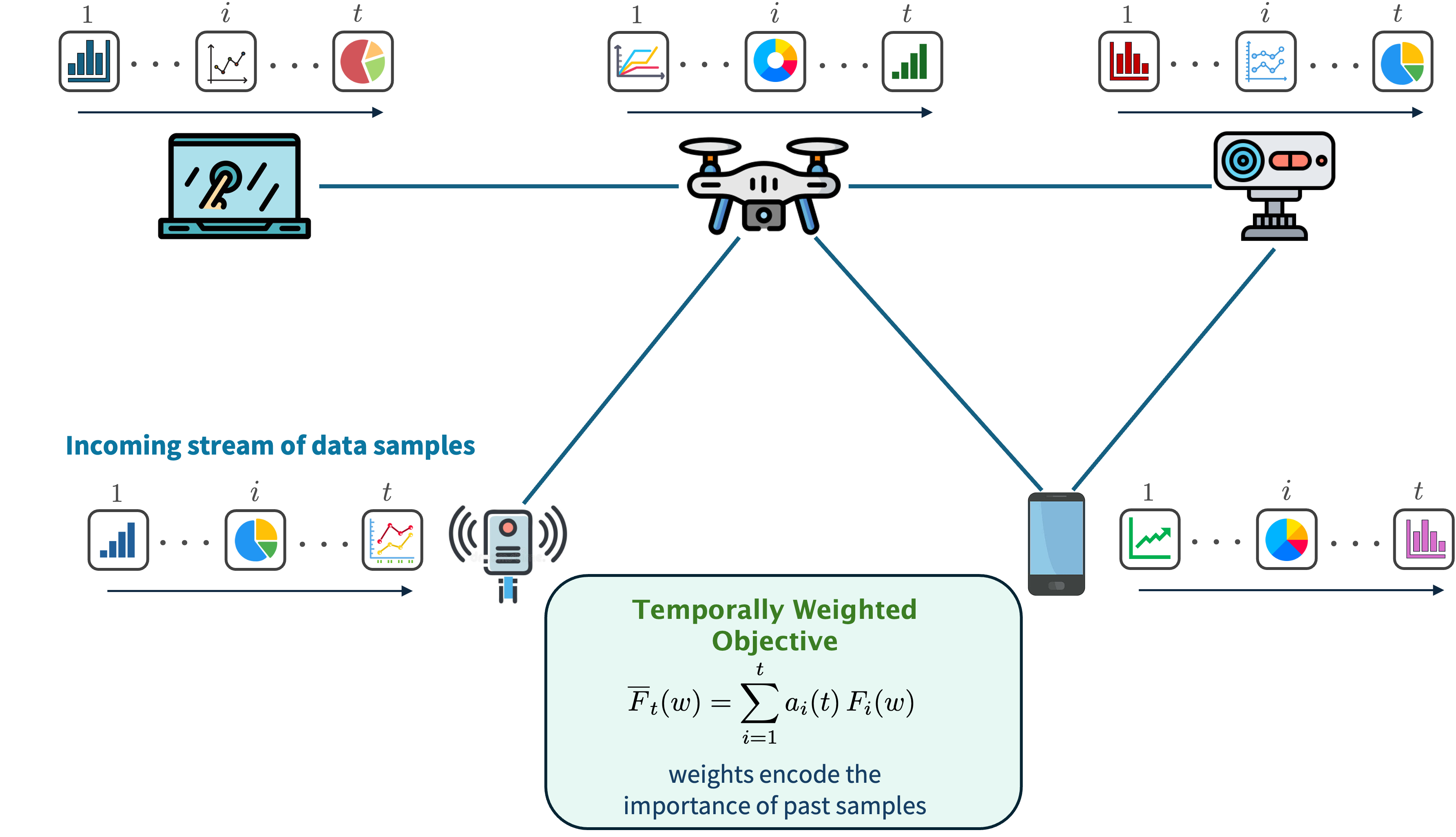}
\caption{Illustration of decentralized learning from streaming data. Each agent receives a local stream of samples and exchanges information with neighboring agents over a communication graph. The network tracks the minimizer of a temporally weighted objective formed from past losses.}
\vspace{-5mm}
\label{fig:system_model}
\end{figure}
\vspace{-3mm}
\section{Temporal Weights for Streaming Losses}
\label{sec:weights}
The temporal weights $\{a_i(t)\}_{i=1}^t$ in \eqref{eq:main_tv_prob} determine how past samples contribute to the current objective and hence control the stability--plasticity tradeoff in streaming environments \cite{CL_Survey}. 
We present two structured ways to generate such weights, and begin with a structured family of temporal weightings induced by kernel sequences, which includes several canonical choices.
\vspace{-2mm}
\subsection{Stationary kernel-induced temporal weights}
\label{subsec:stationary_kernel_weights}
Let $\{g_k\}_{k\ge 0}$ be a nonnegative kernel sequence with $g_0>0$, and define the associated kernel-induced temporal weights for all $t\geq1$ as
\begin{align}
a_i(t) \triangleq \frac{g_{t-i}}{S_t}, \qquad
S_t \triangleq \sum_{k=0}^{t-1}g_k,
\qquad i=1,\dots,t.
\label{eq:kernel_weights_main}
\end{align}
Here, $S_t$ normalizes the weights so the probability simplex constraint \eqref{eq:weight_simplex_constraint} holds. Note that the index
$k=t-i$ can be interpreted as the \emph{age} of sample $i$ at time $t$. Since $a_i(t)\propto g_{t-i}$, the kernel $\{g_k\}$ specifies how sample importance varies with age, and \eqref{eq:kernel_weights_main} implements a simple \emph{shift-and-normalize} weighting rule. 
Several choices follow as special cases:

$\bullet$ \underline{\textit{Uniform Weights:}}
Setting $g_k=1$ for all $k\ge 0$ yields
\begin{align}
a_i(t)=\frac{1}{t},\qquad i=1,\dots,t.
\label{eq:uniform_weights}
\end{align}
This rule is natural in stationary environments.

$\bullet$ \underline{\textit{Exponentially discounted weights: }}Setting $g_k=\gamma^k$ with a discount factor $\gamma\in(0,1)$ results in 
\begin{align}
a_i(t)=\frac{(1-\gamma)\gamma^{t-i}}{1-\gamma^t},\qquad i=1,\dots,t.
\label{eq:exp_disc_weights}
\end{align}
This rule emphasizes recent samples and can be naturally employed for \emph{adaptation} in non-stationary environments; see, e.g., \cite{DallAnese2019_stream} for several applications.

 $\bullet$ \underline{\textit{Windowed-uniform weights:}}
A finite-memory variant of the uniform weighting strategy assigns equal weight only to the most recent $m$ samples, corresponding to $g_k=\mathbbm{1}\{0\le k\le m-1\}$:
\begin{align}
\label{eq:window_uniform_weights}
a_i(t)=
\begin{cases}
\frac{1}{m}, &\quad t-m+1\le i\le t,\\
0, &\quad \text{otherwise},
\end{cases}
\qquad \text{for } t\ge m.
\end{align}

$\bullet$ \underline{\textit{Windowed-discounted weights:}}
Similarly, choosing $g_k=\gamma^k\mathbbm{1}\{0\le k\le m-1\}$ yields
\begin{align}
\label{eq:window_discounted_weights}
a_i(t)=
\begin{cases}
\frac{(1-\gamma)\gamma^{t-i}}{1-\gamma^m}, & t-m+1\le i\le t,\\
0, &\text{otherwise}.
\end{cases}
\qquad \text{for } t\ge m.
\end{align}
This provides a finite-memory approximation of exponentially discounted weighting using only the last $m$ samples.

\vspace{-5mm}
\subsection{Uniform-shrinkage weighting}
\label{subsec:timevarying_kernel}
The stationary kernel family \eqref{eq:kernel_weights_main} assigns weights based on age using a fixed time-invariant kernel sequence. Another possible construction is to generate the weights recursively by shrinking all previous weights by the same factor whenever a new sample arrives.   Let $\{\beta_t\}_{t\ge1}$ satisfy $\beta_1=1$ and $\beta_t\in(0,1)$ for $t\ge2$, and define the time-varying objective function as
\begin{align}
\overline F_{t+1}(\mathbf{w}) = (1-\beta_{t+1})\,\overline F_t(\mathbf{w}) + \beta_{t+1}\,F_{t+1}(\mathbf{w}),
\label{eq:uniform_shrinkage_recursion}
\end{align}
for all $t\geq 1$. Equivalently, the induced weights satisfy
\begin{align}
a_i(t) = \beta_i \prod_{k=i+1}^{t}(1-\beta_k), \qquad i=1,\dots,t.
\label{eq:beta_weights_closed_form}
\end{align}
This family recovers uniform weighting with $\beta_{t+1}=1/(t+1)$ and exponential discounting with $\beta_{t+1}=(1-\gamma)/(1-\gamma^{t+1})$.


Notably, \eqref{eq:beta_weights_closed_form} also allows \emph{time-adaptive} weighting strategies. For example, one may choose $\beta_t$ to decay from an initially constant level toward $1/t$, so that the objective first emphasizes recent samples for fast adaptation and then gradually approaches uniform weighting for stability. Such schedules depend on absolute time, not only on sample age, and are therefore not generally representable by a stationary kernel. On the other hand, since $\beta_t\in(0,1)$ preserves strictly positive weights for all past samples, this construction does not capture the finite-window rules in \eqref{eq:window_uniform_weights}--\eqref{eq:window_discounted_weights}.

\vspace{-3mm}
\section{Algorithm and Performance Metrics}
\label{sec:alg_perf}
We now present the decentralized methods used to solve \eqref{eq:main_tv_prob} and define the associated performance metrics that are adopted throughout the paper. We study two decentralized first-order algorithms: decentralized gradient descent (DGD) \cite{Nedic2009DGD} and adapt-then-combine (ATC) diffusion \cite{AliSayed2014Diffusion}.\footnote{Throughout the paper, ``diffusion'' refers to ATC diffusion, which is closely related to combine-then-adapt diffusion; see \cite{MHT} for a discussion.} Each agent $n$ maintains a local parameter
$\mathbf w_{n,t}\in\mathbb R^d$ at time $t$. Upon receiving new data samples, the objective is updated, after which the agents perform $E \geq 1$ decentralized iterations before the time advances. Each decentralized iteration consists of one mixing and one local gradient step.  This models a fixed communication/computation budget per time step. 

Let 
$\bm{\mathsf w}_t \in \mathbb{R}^{Nd}$ denote the stacked vector of local parameters, $\bm{\mathsf w}_t
    \triangleq
    \begin{bmatrix}
        \mathbf w_{1,t}^{\top} &
        \cdots &
        \mathbf w_{N,t}^{\top}
    \end{bmatrix}^{\top}$ with $\mathbf{w}_{n,t} =[\bm{\mathsf w}_t]_n$. Next, define the block-separable instantaneous objective
\begin{align*}
    f_t(\bm{\mathsf w})
    \triangleq
    \sum_{n=1}^N \ell_{n,t}(\mathbf{w}_n),
\end{align*}
and its temporally weighted version 
\begin{align}
    \overline f_t(\bm{\mathsf w})
    \triangleq
    \sum_{i=1}^t a_i(t)\, f_i(\bm{\mathsf w})
    = \sum_{i=1}^t a_i(t)
    \sum_{n=1}^N \ell_{n,i}(\mathbf{w}_n),
    \label{eq:fbar-def}
\end{align}
whose gradient has the block-stacked form, and is given by
\begin{align}
    \nabla\overline f_t(\bm{\mathsf w})
    {=}
    \sum_{i=1}^t a_i(t)\, \nabla f_i(\bm{\mathsf w})
    {=}
    \sum_{i=1}^t a_i(t)
    \begin{bmatrix}
        \nabla \ell_{1,i}(\mathbf{w}_1) \\
        \vdots \\
         \nabla\ell_{N,i}(\mathbf{w}_N)
    \end{bmatrix}.
    \label{eq:fbar-grad}
\end{align}
The objectives $\overline{F}_t(\cdot)$ in \eqref{eq:main_tv_prob}
and $\overline f_t(\cdot)$ in \eqref{eq:fbar-def} encode the same collection of
sample losses $\{\ell_{n,i}\}$. However, $\overline F_t(\cdot)$ is defined over a common decision parameter $\mathbf{w}\in\mathbb R^d$, whereas
$\overline f_t(\cdot)$ is defined over the \emph{stacked collection} of local parameters $\bm{\mathsf w}\in\mathbb{R}^{Nd}$ which need not be consensual. In particular, on the consensus subspace $\{
    \bm 1_N\otimes\mathbf w:
    \mathbf w\in\mathbb R^d
    \}
    \subseteq\mathbb R^{Nd}$, we have $\overline f_t(\bm 1_N \otimes \mathbf{w})
    = \sum_{i=1}^t a_i(t) \sum_{n=1}^N \ell_{n,i}(\mathbf{w})
    = N\,\overline F_t(\mathbf{w})$.
Thus, on the consensus subspace, the two objectives differ only by the constant factor $N$ and 
and their minimizers correspond
through the mapping
$\mathbf w\mapsto\bm 1_N\otimes\mathbf w$.
Finally, we let
$ \overline{\bm{\mathsf w}}_t^*
 \triangleq
    \bm 1_N\otimes\overline{\mathbf w}_t^* \in\mathbb R^{Nd}
$
denote the stacked version of the global minimizer in \eqref{eq:main_tv_prob}.


After the new samples arrive at time $t+1$, the agents perform
$E$ decentralized iterations on the updated objective. For DGD, the
$k$-th iteration at agent $n$ is
\begin{align}
\mathbf w_{n,t}^{(k+1)}
{=}
\sum_{m=1}^{N}
[\mat M]_{n,m}\mathbf w_{m,t}^{(k)}
{-}
{\eta}
\sum_{i=1}^{t+1}a_i(t+1)\nabla\ell_{n,i}
\bigl(\mathbf w_{n,t}^{(k)}\bigr),
\label{eq:dgd_agent_update}
\end{align}
whereas diffusion performs the adaptation and combination:
\begin{align}
\widetilde{\mathbf w}_{n,t}^{(k+1)}
&=
\mathbf w_{n,t}^{(k)}
-
\eta
\sum_{i=1}^{t+1}a_i(t+1)\nabla\ell_{n,i}
\bigl(\mathbf w_{n,t}^{(k)}\bigr),
\label{eq:diffusion_adapt}
\\
\mathbf w_{n,t}^{(k+1)}
&=
\sum_{m=1}^{N}
[\mat M]_{n,m}
\widetilde{\mathbf w}_{m,t}^{(k+1)},
\label{eq:diffusion_combine}
\end{align}
for $k=0,\ldots,E-1$. Stacking the local parameters across the network, the DGD and diffusion
updates can be expressed in the unified form
\begin{align}
\bm{\mathsf w}_t^{(k+1)}
= \overline{\mat{M}}\bm{\mathsf w}_t^{(k)}- \eta \overline{\mat{Z}} \nabla \overline f_{t+1}(\bm{\mathsf w}_t^{(k)}),
\label{eq:unified_inner}
\end{align}
where $\eta>0$ is a constant step size, and we define $\overline{\mat M}
    \triangleq
    \mat M\otimes\mat I_d$, $\overline{\mat Z}
    \triangleq
    \mat Z\otimes\mat I_d \in \mathbb R^{Nd\times Nd}$ and
\begin{align}
\mat Z \triangleq
\begin{cases}
\mat I, & \text{DGD},\\
\mat M, & \text{diffusion}.
\end{cases}
\label{eq:Z_def}
\end{align}
For both methods,
$\bm{\mathsf w}_t^{(0)}=\bm{\mathsf w}_t$, and the macro-update (the stacked local parameter update) after $E$ iterations is 
\begin{align}
\bm{\mathsf w}_{t+1}
    \triangleq
    \bm{\mathsf w}_t^{(E)}.
\label{eq:macro_update}
\end{align}
When $E=1$, \eqref{eq:unified_inner}--\eqref{eq:macro_update} reduce to the standard single-step DGD/diffusion updates applied to the time-varying objective.

In the considered streaming setting, the objective changes whenever new data arrive; therefore, the stacked iterates $\{\bm{\mathsf w}_t\}$ generally do not converge.
We therefore measure performance through the tracking error
\begin{align}
    \mathrm{TE}(t)
    \triangleq
    \left\|
    \bm{\mathsf w}_t - \overline{\bm{\mathsf w}}_t^*
    \right\|
    =
    \left\|
    \bm{\mathsf w}_t
    -
    \bm 1_N\otimes\overline{\mathbf w}_t^*
    \right\|,
    \label{eq:TE_define}
\end{align}
and the asymptotic tracking error
\begin{align}
    \mathrm{ATE}
    \triangleq
    \limsup_{t\to\infty}
    \left\|
    \bm{\mathsf w}_t - \overline{\bm{\mathsf w}}_t^*
    \right\|.
    \label{eq:ATE_define}
\end{align}
\vspace{-7mm}
\section{Tracking Error Analysis}
\label{sec:analysis}
In this section, we analyze the decentralized updates \eqref{eq:unified_inner}--\eqref{eq:macro_update} to solve the streaming time-varying problem \eqref{eq:main_tv_prob}. Following \cite{MHT}, we adopt a contraction-mapping viewpoint, which enables a {unified} TE analysis for both DGD and diffusion.
We first derive generic TE bounds and then specialize them to the temporal weighting strategies in Section~\ref{sec:weights}.
Define the time-varying decentralized update operator
$\phi_t{:}\mathbb R^{Nd}{\to}\mathbb R^{Nd}$ as
\begin{align}
\phi_t(\bm{\mathsf w})
\triangleq
\overline{\mat M}\bm{\mathsf w}
-
\eta\overline{\mat Z}
\nabla\overline f_t(\bm{\mathsf w}),
\label{eq:phi-def-unified}
\end{align}
where $\overline{\mat M} =
    \mat M\otimes\mat I_d$ and $\overline{\mat Z} = \mat Z\otimes\mat I_d$. The decentralized updates in \eqref{eq:unified_inner}--\eqref{eq:macro_update} with $E$ iterations at time $t$ can then be written as
\begin{align}
\bm{\mathsf w}_{t+1}
=
(\underbrace{\phi_{t+1}\circ\cdots\circ\phi_{t+1}}_{E\ \text{times}})
(\bm{\mathsf w}_t)
\triangleq
\Phi_{t+1}(\bm{\mathsf w}_t),
\label{eq:Phi-def-unified}
\end{align}
where we use $\Phi_t(\cdot)$ to denote the composition of the mapping $\phi_t(\cdot)$ applied $E$ times.
For the analysis, we make the following standard assumptions;
see, e.g., \cite{Nesterov,boyd2004convex,MHT}.
\begin{Ass}
\label{ass:smooth_sc}
For every agent $n$ and time $t$, the loss $\ell_{n,t}(\cdot)$ is $L$-smooth and $\mu$-strongly convex. Then, for all
$\mathbf x,\mathbf y\in\mathbb R^d$, 
\begin{align}
    \mu\|\mathbf x-\mathbf y\|\leq\left\|
\nabla\ell_{n,t}(\mathbf x)
-
\nabla\ell_{n,t}(\mathbf y)
\right\|
\le
L\|\mathbf x-\mathbf y\|
\end{align}

\end{Ass}
We use $\kappa \triangleq L/\mu$ with $\kappa\geq 1$ to denote the condition number of the optimization problem. Since $F_t(\cdot)$ and $\overline{F}_t(\cdot)$ are convex combinations of the losses at time $t$ and up to time $t$, respectively, both are also $L$-smooth and $\mu$-strongly convex for all $t\geq 1$.
The same properties also hold for the block-separable objectives $f_t(\cdot)$ and $\overline{f}_t(\cdot)$. That is, for all $\bm{\mathsf u},\bm{\mathsf v}\in\mathbb R^{Nd}$ and for either choice
$\tilde f_t \in \{f_t,\overline{f}_t\}$,
\begin{align}
\bigl\|\nabla \tilde f_t(\bm{\mathsf u}) - \nabla \tilde f_t(\bm{\mathsf v})\bigr\|
   \leq L \|\bm{\mathsf u}-\bm{\mathsf v}\|,
\label{eq:tilde_f_t_smooth}
\end{align}
\begin{align}
\tilde f_t(\bm{\mathsf v})
   \ge \tilde f_t(\bm{\mathsf u})
      +  \nabla \tilde f_t(\bm{\mathsf u})^\top (\bm{\mathsf v}-\bm{\mathsf u})
      + \frac{\mu}{2} \|\bm{\mathsf v}-\bm{\mathsf u}\|^2.
\label{eq:tilde_f_t_sc}
\end{align}

Under Assumption \ref{ass:smooth_sc}, the analysis in
Section IV.C of \cite{MHT} gives the following lemma.

\begin{lemma}
\label{lem:contraction-unified}
Suppose Assumption~\ref{ass:smooth_sc} holds. Then, $\phi_t(\cdot)$ is a contraction for every $t$, and
\begin{align}
\|\phi_t(\bm{\mathsf u})-\phi_t(\bm{\mathsf v})\|
\le
(1-\eta\mu)\|\bm{\mathsf u}-\bm{\mathsf v}\|,
\quad
\forall\,\bm{\mathsf u},\bm{\mathsf v}\in\mathbb R^{Nd}.
\end{align}
under the step size conditions 
$0<\eta\le (1+\lambda_N)/(L+\mu)$ for DGD and
$0<\eta\le 2/(L+\mu)$ for diffusion.
Consequently, $\Phi_t$ is a contraction with factor $\alpha \triangleq (1-\eta\mu)^E\in(0,1)$. 
\end{lemma}

Throughout the rest of the analysis, we assume the step size
conditions of Lemma~\ref{lem:contraction-unified}. Hence, by
Banach's fixed-point theorem~\cite{Rudin1976}, each $\phi_t(\cdot)$ admits a unique
fixed point, given by
\begin{align}
\widetilde{\bm{\mathsf w}}_t
\in\mathbb R^{Nd}
\quad\text{s.t.}\quad
\widetilde{\bm{\mathsf w}}_t
=
\phi_t(\widetilde{\bm{\mathsf w}}_t).
\label{eq:fixed-point-def-unified}
\end{align}
Since $\Phi_t(\cdot)$ is the $E$-fold composition of $\phi_t(\cdot)$,  $\widetilde{\bm{\mathsf w}}_t=\Phi_t(\widetilde{\bm{\mathsf w}}_t)$. The sequence $\{\widetilde{\bm{\mathsf w}}_t\}$ provides a natural reference trajectory for the analysis: it represents the network state to which the decentralized method would converge if the objective
$\overline f_t(\cdot)$ were frozen at time $t$.
Leveraging this perspective, the tracking error in \eqref{eq:TE_define} can be upper-bounded as 
\begin{align}
    \mathrm{TE}(t) &= \bigl\|
    \bm{\mathsf w}_t - \overline{\bm{\mathsf w}}_t^*
    \bigr\|\leq
    \underbrace{\bigl\|\bm{\mathsf w}_t- \widetilde{\bm{\mathsf w}}_t \bigr\|}_{\mathrm{FPTE}(t)}
    +
    \underbrace{\bigl\|\widetilde{\bm{\mathsf w}}_t - \overline{\bm{\mathsf w}}_t^*\bigr\|}_{\text{FP bias}}.
    \label{eq:TE-decomposition}
\end{align}
The \emph{fixed-point tracking error} (FPTE) above captures the ability of the algorithms to track a moving fixed point under a limited per-time-step iteration budget. The second term is the fixed-point bias, capturing the distance between the fixed point and the minimizer at time $t$.

 To bound these two terms, we next impose the following boundedness condition on the sample-wise minimizers.
\begin{Ass}
\label{ass:bounded_minimizers}
For each agent $n$ and time $t$, let $\mathbf w_{n,t}^*\in
\arg\min_{\mathbf w\in\mathbb R^d}
\ell_{n,t}(\mathbf w)$. There exists $C>0$ such that $\|\mathbf w_{n,t}^*\|\le C$ for all $n$ and $t$.
\end{Ass}
Assumption~\ref{ass:bounded_minimizers} is a structural condition on the sample losses, rather than a direct bounded-drift assumption on the minimizer sequence, as commonly used in time-varying optimization; see, e.g., \cite{Time_Structured}. Similar bounded-minimizer conditions have appeared in streaming and federated learning analyses \cite{TV_asilomar25,ChungHu25StreamFL}. Here, it ensures that the time-varying minimizer $\{\overline {\mathbf{w}}_t^*\}$ and the fixed-point sequences $\{\widetilde{\bm{\mathsf w}}_t\}$ are uniformly bounded.

\begin{lemma}
\label{lem:boundedness_compact}
Under Assumptions~\ref{ass:smooth_sc} and~\ref{ass:bounded_minimizers},
it holds that
\begin{align}
\|\overline{\mathbf w}_t^*\|
\le C\sqrt{\kappa}..
\label{eq:global-min-bnd}
\end{align}
Moreover, the fixed points $\{\widetilde{\bm{\mathsf w}}_t\}$ satisfy, for all $t\ge 1$,
\begin{align}
&\|\widetilde{\bm{\mathsf w}}_t\|\le \sqrt{N}\,C_Z, \text{ where }\nonumber\\&
C_Z \triangleq
\begin{cases}
C\sqrt{\kappa}, & \mat Z=\mat I \quad \text{(DGD)},\\
C\kappa, & \mat Z=\mat M \quad \text{(diffusion)},\\
C\sqrt{\kappa}, & \mat Z=\mat M,\ \mat M\succeq 0. 
\end{cases}
\label{eq:fixedpoint-bnd-compact}
\end{align}
Furthermore, for all $i,t\ge 1$,
\begin{align}
\bigl\|
\nabla\overline f_t(\widetilde{\bm{\mathsf w}}_i)
\bigr\|,
\;
\left\|
\nabla f_t(\widetilde{\bm{\mathsf w}}_i)
\right\|,
\;
\left\|
\nabla\overline f_t(\overline{\bm{\mathsf w}}_t^*)
\right\|
\le G_Z.
\label{eq:grad-bnd-G}
\end{align}
where  $G_Z \triangleq 2L\sqrt{N}\,C_Z$.
\end{lemma}

\noindent The proof of Lemma \ref{lem:boundedness_compact} can be found in Appendix. We note that the condition $\mat M\succeq0$ is not required by the algorithms or by the TE analysis; it only yields a tighter bound in (\ref{eq:fixedpoint-bnd-compact})  for diffusion.


\subsection{Fixed-point bias}
\label{subsec:bias_term}
To bound the TE, we first bound the bias term in \eqref{eq:TE-decomposition}. The following result is obtained by applying the fixed-point-to-minimizer bound of \cite[App.~A]{MHT} to $\overline f_t(\cdot)$ and then using Lemma~\ref{lem:boundedness_compact}.
\begin{prop}
\label{prop:bias_bound}
Under the conditions above,
\begin{align}
\left\|
\widetilde{\bm{\mathsf w}}_t
-
\overline{\bm{\mathsf w}}_t^*
\right\|
&\le
\eta\kappa\Lambda_Z
\left\|
\nabla\overline f_t
(\overline{\bm{\mathsf w}}_t^*)
\right\|\le
\eta\kappa\Lambda_ZG_Z,
\label{eq:bias_bound}
\end{align}
where
\begin{align}
\Lambda_Z \triangleq
\begin{cases}
\|(\mat I-\mat M)^\dagger\|_2=\dfrac{1}{1-\lambda_2}\;,\;
\mat Z=\mat I \quad \text{(DGD)},\\[3mm]
2\|(\mat{I}-\mat{M})^\dagger \mat{M}\|_2, \qquad\mat{Z}=\mat{M}
\; \text{(diffusion)}.
\end{cases}
\label{eq:topology_factor_bias}
\end{align}

For diffusion, the first inequality additionally requires $\eta\le 1/(L\Lambda_Z)$.
\end{prop}
The bounds in 
\eqref{eq:bias_bound} show that the fixed point bias scales linearly with the step size $\eta$ and depends on the condition number $\kappa$, the topology factor $\Lambda_Z$, and the gradient magnitude $\|\nabla\overline f_t
(\overline{\bm{\mathsf w}}_t^*)\|$, which reflects data heterogeneity across agents. In particular, in homogeneous settings where all agents observe identical losses, they share the same minimizer, and hence $\|\nabla\overline f_t
(\overline{\bm{\mathsf w}}_t^*)\| =0$, implying zero bias. 
Better network connectivity reduces $\Lambda_Z$; for the complete graph $\mat M=\frac{1}{N}\mathbf 1\mathbf 1^\top$, one obtains $\Lambda_Z=1$ for DGD and $\Lambda_Z=0$ for diffusion.

\subsection{Fixed-point tracking error and drift}
\label{subsec:fpte_drift}

We next bound the fixed-point tracking term in \eqref{eq:TE-decomposition}. Since $\bm{\mathsf w}_{t+1}
=
\Phi_{t+1}(\bm{\mathsf w}_t)$ and $\widetilde{\bm{\mathsf w}}_{t+1}=\Phi_{t+1}\bigl(\widetilde{\bm{\mathsf w}}_{t+1}\bigr)$, the contraction of $\Phi_{t+1}(\cdot)$ from Lemma~\ref{lem:contraction-unified} yields $\mathrm{FPTE}(t+1)$
\begin{align}
&= \|\bm{\mathsf w}_{t+1} - \widetilde{\bm{\mathsf w}}_{t+1} \| \nonumber\\
&=\|\Phi_{t+1}(\bm{\mathsf w}_t)-\Phi_{t+1}(\widetilde{\bm{\mathsf w}}_{t+1})\| \nonumber\\
&\le \alpha \|\bm{\mathsf w}_t-\widetilde{\bm{\mathsf w}}_{t+1}\| \le \alpha \mathrm{FPTE}(t)
+\alpha \|\widetilde{\bm{\mathsf w}}_{t+1}-\widetilde{\bm{\mathsf w}}_t\|,
\label{eq:FPTE-recursion}
\end{align}
where recall that $\alpha = (1-\eta\mu)^E$. 
Unrolling \eqref{eq:FPTE-recursion} and using the same contraction argument yields
\begin{align}
\mathrm{FPTE}(t)
\le
\alpha^t\|\bm{\mathsf w}_0-\widetilde{\bm{\mathsf w}}_1\|
+
\sum_{i=1}^{t-1}
\alpha^{t-i}
\|\widetilde{\bm{\mathsf w}}_{i+1}-\widetilde{\bm{\mathsf w}}_i\|.
\label{eq:FPTE-unrolled-main}
\end{align}
In \eqref{eq:FPTE-unrolled-main}, the first term decreases geometrically to 
zero as $t\to\infty$, while the second term accumulates the drift of the moving fixed-point $\{ \widetilde{\bm{\mathsf w}}_t\}$. Thus, controlling the FPTE reduces to controlling this drift. We next bound this drift. 
\begin{lemma}
\label{lem:generic-drift-unified}
Under Assumption~\ref{ass:smooth_sc}, the fixed points $\{\widetilde{\bm{\mathsf w}}_t\}$ satisfy
\begin{align}
\|\widetilde{\bm{\mathsf w}}_{t+1}-\widetilde{\bm{\mathsf w}}_t\|
\le
\frac{1}{\mu}
\left\|
\nabla\overline f_{t+1}(\widetilde{\bm{\mathsf w}}_{t+1})
-
\nabla\overline f_t(\widetilde{\bm{\mathsf w}}_{t+1})
\right\|.
\label{eq:generic-drift-unified}
\end{align}
\end{lemma}
\noindent The proof is given in Appendix. Lemma~\ref{lem:generic-drift-unified} relates the fixed-point drift to the temporal variation of the gradient of the time-weighted objective. In particular, if $\nabla \overline f_t(\cdot)$ changes slowly with $t$, then the corresponding fixed-point sequence $\{\widetilde{\bm{\mathsf w}}_t\}$ cannot drift rapidly. 

Next, using
$
\nabla\overline f_{t+1}(\bm{\mathsf w})-\nabla\overline f_t(\bm{\mathsf w})
=
a_{t+1}(t+1)\nabla f_{t+1}(\bm{\mathsf w})
+
\sum_{i=1}^t
\big(a_i(t+1)-a_i(t)\big)\nabla f_i(\bm{\mathsf w})$,
the gradient bound in Lemma~\ref{lem:boundedness_compact} and applying  the triangle inequality, \eqref{eq:generic-drift-unified} implies
\begin{align}
\|\widetilde{\bm{\mathsf w}}_{t+1}-\widetilde{\bm{\mathsf w}}_t\|
\le
\frac{G_Z}{\mu}
\Big(
a_{t+1}(t+1)
+
\sum_{i=1}^t |a_i(t+1)-a_i(t)|
\Big).
\label{eq:fp_drift_weight_variation}
\end{align}
The following identity further simplifies \eqref{eq:fp_drift_weight_variation}; its proof is given in the Appendix. 
\begin{lemma}
\label{lem:weight_variation_identity}
Let
$
I_t \triangleq \sum_{i=1}^t \bigl(a_i(t+1)-a_i(t)\bigr)_+$.
Then, 
\begin{align}
\sum_{i=1}^t |a_i(t+1)-a_i(t)|
= a_{t+1}(t+1)+2I_t.
\label{eq:weight_variation_identity}
\end{align}
\end{lemma}
Leveraging Lemma~\ref{lem:weight_variation_identity}, the fixed-point drift in \eqref{eq:fp_drift_weight_variation} can be rewritten as
\begin{align}
\|\widetilde{\bm{\mathsf w}}_{t+1}-\widetilde{\bm{\mathsf w}}_t\|
\le
\frac{2G_Z}{\mu}
\big(a_{t+1}(t+1)+I_t\big).
\label{eq:fp_drift_gen_It}
\end{align}
Moreover, if the weights are non-increasing over time, i.e.,
$ a_i(t+1)\le a_i(t),  \forall i\le t$, then $I_t=0$, and \eqref{eq:fp_drift_gen_It} simplifies to
\begin{align}
\|\widetilde{\bm{\mathsf w}}_{t+1}-\widetilde{\bm{\mathsf w}}_t\|
\le \frac{2G_Z}{\mu}\,a_{t+1}(t+1).
\label{eq:fp_drift_monotone_old_weights}
\end{align}
This property is satisfied by the uniform-shrinkage family in Section~\ref{subsec:timevarying_kernel}; it is also satisfied by the stationary-kernel family in Section~\ref{subsec:stationary_kernel_weights} whenever the kernel sequence $\{g_k\}$ is non-increasing. Therefore, it holds for the four canonical weighting strategies considered in Section~\ref{sec:weights}, whose drift bounds are summarized next.

\begin{prop}
\label{prop:drift_canonical_weights}
Under \eqref{eq:fp_drift_monotone_old_weights}, the following bounds hold.

\noindent $\bullet$ \underline{\textit{Uniform weights in \eqref{eq:uniform_weights}}:}  
\begin{align}
\|\widetilde{\bm{\mathsf w}}_{t+1}-\widetilde{\bm{\mathsf w}}_t\|
\le \frac{2G_Z}{\mu(t+1)}.
\label{eq:drift_uniform_simple}
\end{align}

\noindent $\bullet$ \underline{\textit{Exponentially discounted weights in \eqref{eq:exp_disc_weights}}:}
\begin{align}
\|\widetilde{\bm{\mathsf w}}_{t+1}-\widetilde{\bm{\mathsf w}}_t\|
\le \frac{2G_Z}{\mu}\frac{1-\gamma}{1-\gamma^{t+1}}.
\label{eq:drift_discounted_simple}
\end{align}

\noindent $\bullet$ \underline{\textit{Windowed-uniform weights in \eqref{eq:window_uniform_weights}}:}  
For all $t\ge m$,
\begin{align}
\|\widetilde{\bm{\mathsf w}}_{t+1}-\widetilde{\bm{\mathsf w}}_t\|
\le \frac{2G_Z}{\mu m}.
\label{eq:drift_window_uniform_simple}
\end{align}

\noindent $\bullet$ \underline{\textit{Windowed-discounted weights in \eqref{eq:window_discounted_weights}}:}  For all $t\ge m$,
\begin{align}
\|\widetilde{\bm{\mathsf w}}_{t+1}-\widetilde{\bm{\mathsf w}}_t\|
\le \frac{2G_Z}{\mu}\frac{1-\gamma}{1-\gamma^m}.
\label{eq:drift_window_discounted_simple}
\end{align}
\end{prop}
Substituting these drift bounds into \eqref{eq:FPTE-unrolled-main}, and combining with the bias bound in \eqref{eq:bias_bound}, yields the TE guarantees discussed next.
\subsection{Tracking error guarantees}
\label{subsec:te_guarantees}
We next present the TE bounds specialized to the four canonical weighting strategies.
All bounds are stated under the step size conditions of Lemma~\ref{lem:contraction-unified} and, for diffusion, the additional condition in Proposition~\ref{prop:bias_bound}.

\subsubsection*{Uniform temporal weights}
For the uniform temporal weighting strategy \eqref{eq:uniform_weights}, utilizing \eqref{eq:drift_uniform_simple}, the following holds
\begin{align}
\mathrm{TE}(t)
&\le \alpha^{t}\,\|\bm{\mathsf w}_0-\widetilde{\bm{\mathsf w}}_1\|
+ \frac{2G_Z}{\mu}\sum_{i=1}^{t-1}\frac{\alpha^{t-i}}{i+1} + \eta\kappa\Lambda_ZG_Z.
\label{eq:TE_uniform}
\end{align}
The summation term in \eqref{eq:TE_uniform} admits an $\mathcal O(1/t)$ upper bound, 
yielding the following theorem.


\begin{theorem}
    \label{thm:TE_uniform}
    Define $S(t)\triangleq\sum_{i=1}^{t-1}\frac{\alpha^{t-i}}{i+1}$ with $\alpha = (1-\eta\mu)^E \in(0,1)$. Under uniform weights \eqref{eq:uniform_weights}, the tracking error satisfies for all $t \geq t_0\triangleq \bigl\lceil \frac{2\alpha}{1-\alpha}\bigr\rceil$
\begin{align}
\mathrm{TE}(t)
\le \alpha^t\|\bm{\mathsf w}_0-\widetilde{\bm{\mathsf w}}_1\|
+\frac{2G_Z}{\mu}\frac{A}{t} + \eta\kappa\Lambda_Z G_Z,
\label{eq:TE_uniform_final}
\end{align}
where $A\triangleq \max\left\{t_0S(t_0),\frac{2\alpha}{1-\alpha}\right\}$.
Furthermore,
\begin{align}
    \limsup_{t \to \infty}
    \mathrm{TE}(t)=\mathrm{ATE} \le \eta\kappa\Lambda_ZG_Z.
    \label{eq:ATE_uniform_final}
\end{align}
\end{theorem}
\noindent The proof is given in Appendix. Since $\alpha\in(0,1)$, the first term in \eqref{eq:TE_uniform_final} decays geometrically, while the second term decays as $\mathcal O(1/t)$. Thus, under uniform weighting, the FPTE contribution vanishes asymptotically, 
and the TE is ultimately governed solely by the fixed-point bias term, which is $\mathcal{O}(\eta)$. In particular, to guarantee $\mathrm{ATE}=\limsup_{t \to \infty} 
    \mathrm{TE}(t)\le\epsilon$, it suffices to choose
\begin{align}
    \eta \le \frac{\epsilon}{\kappa\Lambda_ZG_Z}.
    \label{eq:eta_uniform_eps}
\end{align}

\subsubsection*{Exponentially discounted weights}

For the exponentially discounted temporal weighting strategy \eqref{eq:exp_disc_weights}, utilizing \eqref{eq:drift_discounted_simple} gives
\begin{align}
\mathrm{TE}(t)
&\le \alpha^{t}\|\bm{\mathsf w}_0-\widetilde{\bm{\mathsf w}}_1\|
+\frac{2G_Z}{\mu}
\sum_{i=1}^{t-1}
\alpha^{t-i}\frac{1-\gamma}{1-\gamma^{i+1}}
+\eta\kappa\Lambda_ZG_Z.
\label{eq:TE_discounted_pre}
\end{align}


The summation term in \eqref{eq:TE_discounted_pre} admits a non-vanishing asymptotic bound, yielding the following theorem.
\begin{theorem}
\label{thm:TE_discounted}
 Define $S_\gamma(t) \triangleq \sum_{i=1}^{t-1}  \frac{(1-\gamma)\alpha^{t-i}}{1 - \gamma^{i+1}}$ with $\alpha = (1-\eta\mu)^E \in(0,1)$. Under 
 discounted weights \eqref{eq:exp_disc_weights}, the tracking error satisfies for all $t \geq t_0 \triangleq \lceil\ln(\frac{1-\alpha}{1+\alpha-2\gamma\alpha})/\ln(\gamma)\rceil$
\begin{align}
\mathrm{TE}(t)
\le
\alpha^t\|\bm{\mathsf w}_0-\widetilde{\bm{\mathsf w}}_1\|
+
\frac{2G_Z}{\mu}A_\gamma\frac{1-\gamma}{1-\gamma^t}
+
\eta\kappa\Lambda_ZG_Z,
\label{eq:TE_discounted_final}
\end{align}
where $A_\gamma\triangleq\max\bigl\{\frac{(1-\gamma^{t_0})S_\gamma(t_0)}{1-\gamma}, \frac{2\alpha}{1-\alpha}\bigr\}$. Furthermore,
\begin{align}
\mathrm{ATE}
\le
\frac{2G_Z}{\mu}\frac{(1-\gamma)\alpha}{1-\alpha}
+
\eta\kappa\Lambda_ZG_Z.
\label{eq:ATE_discounted_final}
\end{align}
\end{theorem}
\noindent The proof is given in Appendix. Unlike uniform weighting, exponentially discounted weighting induces a \emph{non-vanishing} FPTE contribution in addition to the bias floor. This is because old samples are exponentially forgotten, 
and, hence, the fixed point continues to drift by an amount proportional to $1-\gamma$ (cf. \eqref{eq:drift_discounted_simple}). As $\gamma\to1$, this additional floor vanishes and the discounted rule approaches uniform weighting.

For a target $\epsilon>0$, suppose we choose $0<\eta\le \epsilon/(2\kappa\Lambda_ZG_Z)$ so that the bias term $\eta\kappa\Lambda_ZG_Z$ is at most $\epsilon/2$. Then, from \eqref{eq:ATE_discounted_final}, it is sufficient to choose
\begin{align}
E \ge
\vast\lceil
\frac{\ln\!\Big(\frac{\epsilon}{\epsilon+\frac{4G_Z}{\mu}(1-\gamma)}\Big)}
{\ln(1-\eta\mu)}
\vast\rceil,
\label{eq:E_disc_eps}
\end{align}
to guarantee $\mathrm{ATE}\le\epsilon$.
\subsubsection*{Windowed temporal weights}
We next consider the finite-memory windowed variants. Unlike the full-memory weightings, windowed strategies retain only a fixed number of past samples, thereby reducing memory and computational requirements. This finite-memory constraint, however, generally induces a non-vanishing FPTE. We first state the result for windowed-discounted weights; the windowed-uniform case is then recovered as the limiting case $\gamma\to1$. 
Using the drift bound in \eqref{eq:drift_window_discounted_simple}, together with 
the geometric sum $\sum_{i=1}^{t-1}\alpha^{t-i}= \alpha(1-\alpha^{t-1})/(1-\alpha)$, we obtain the following bound.

\begin{theorem}
\label{thm:TE_window_discounted}
Under windowed-discounted weights \eqref{eq:window_discounted_weights}, for all $t\ge m$, the tracking error satisfies
\begin{align}
\mathrm{TE}(t)
&\le
\alpha^t\|\bm{\mathsf w}_0-\widetilde{\bm{\mathsf w}}_1\|
+
\frac{2G_Z}{\mu}
\frac{1-\gamma}{1-\gamma^m}
\frac{\alpha(1-\alpha^{t-1})}{1-\alpha}\nonumber\\
&\qquad+ \eta\kappa\Lambda_ZG_Z.
\label{eq:TE_window_discounted_final}
\end{align}
Consequently, 
\begin{align}
\mathrm{ATE}
\le
\frac{2G_Z}{\mu}
\frac{1-\gamma}{1-\gamma^m}
\frac{\alpha}{1-\alpha}
+
\eta\kappa\Lambda_ZG_Z.
\label{eq:ATE_window_discounted_final}
\end{align}
\end{theorem}
For windowed-discounted weights, the FPTE contribution in
\eqref{eq:ATE_window_discounted_final} is controlled by the
effective-memory factor $(1-\gamma)/(1-\gamma^m)$. As $m\to\infty$,
this factor approaches $1-\gamma$, matching the bound in \eqref{eq:ATE_discounted_final} for full-memory exponentially discounted weighting.

For a target $\epsilon>0$, suppose $\eta\le \epsilon/(2\kappa\Lambda_ZG_Z)$. Then, from \eqref{eq:ATE_window_discounted_final}, it is sufficient to choose
\begin{align}
E \ge
\vast\lceil
\frac{\ln\!\Big(\frac{\epsilon}{\epsilon+\frac{4G_Z}{\mu}\cdot\frac{1-\gamma}{1-\gamma^m}}\Big)}
{\ln(1-\eta\mu)}
\vast\rceil,
\label{eq:E_wd_eps}
\end{align}
to guarantee $\mathrm{ATE}\le\epsilon$.

Since $\lim_{\gamma\to1}(1-\gamma)/(1-\gamma^m)=1/m$, the windowed-uniform guarantees follow directly as a special case.

\begin{corollary}
\label{cor:TE_window_uniform}
Under windowed-uniform weights \eqref{eq:window_uniform_weights}, for all $t\ge m$,
\begin{align}
\mathrm{TE}(t)
\le
\alpha^t\|\bm{\mathsf w}_0-\widetilde{\bm{\mathsf w}}_1\|
+
\frac{2G_Z}{\mu m}
\frac{\alpha(1-\alpha^{t-1})}{1-\alpha}
+
\eta\kappa\Lambda_ZG_Z.
\label{eq:TE_window_uniform_final}
\end{align}
Consequently,
\begin{align}
\mathrm{ATE}
\le
\frac{2G_Z}{\mu m}\frac{\alpha}{1-\alpha}
+
\eta\kappa\Lambda_ZG_Z.
\label{eq:ATE_window_uniform_final}
\end{align}
Moreover, for a target $\epsilon>0$, if $0<\eta\le \epsilon/(2\kappa\Lambda_ZG_Z)$, it is sufficient to choose 
\begin{align}
E \ge
\vast\lceil
\frac{\ln\!\Big(\frac{\epsilon}{\epsilon+\frac{4G_Z}{\mu m}}\Big)}
{\ln(1-\eta\mu)}
\vast\rceil
\label{eq:E_wu_eps}
\end{align}
to guarantee $\mathrm{ATE}\le\epsilon$.
\end{corollary}
For fixed $\eta$ and $E$, the finite-memory FPTE contribution in
\eqref{eq:ATE_window_uniform_final} has a $1/m$ prefactor, whereas the overall ATE also includes the fixed-point bias term. Hence, increasing the window length reduces the additional finite-memory contribution. In the limit $m\to\infty$, this contribution vanishes, and the bound reduces to the 
ATE bound obtained under full-memory uniform weighting in \eqref{eq:ATE_uniform_final}.

The iteration complexity needed to achieve $\mathrm{ATE}\le\epsilon$ is summarized next.
\begin{remark}
Under uniform weighting, the FPTE contribution vanishes for any $E\ge1$, and the ATE in \eqref{eq:ATE_uniform_final} is controlled solely by the bias term. Hence, the condition in \eqref{eq:eta_uniform_eps} shows that the ATE can be made arbitrarily small by reducing the step size $\eta$. In contrast, for the discounted and windowed strategies, the ATE contains a non-vanishing FPTE floor in addition to the bias term. The sufficient conditions in \eqref{eq:E_disc_eps}, \eqref{eq:E_wd_eps}, and \eqref{eq:E_wu_eps}
reveal a tradeoff between the step size $\eta$ and the number of decentralized iterations per time-step $E$. For fixed $\eta$, the sufficient threshold on $E$ for achieving an FPTE floor of order $\epsilon$ scales as $\Theta(\ln(1/\epsilon))$. However, achieving an arbitrarily small total ATE also requires reducing the bias term, which motivates choosing $\eta=\Theta(\epsilon)$. In this case, $-\ln(1-\eta\mu)=\Theta(\epsilon)$, and the sufficient per-step iteration budgets in \eqref{eq:E_disc_eps}--\eqref{eq:E_wd_eps} scale as $\Theta(\ln(1/\epsilon)/\epsilon)$.
\end{remark}
\vspace{-4mm}

\section{Numerical Results}
\label{sec:numerics}

\begin{figure}[t]
    \centering
    \includegraphics[width=0.90\linewidth]{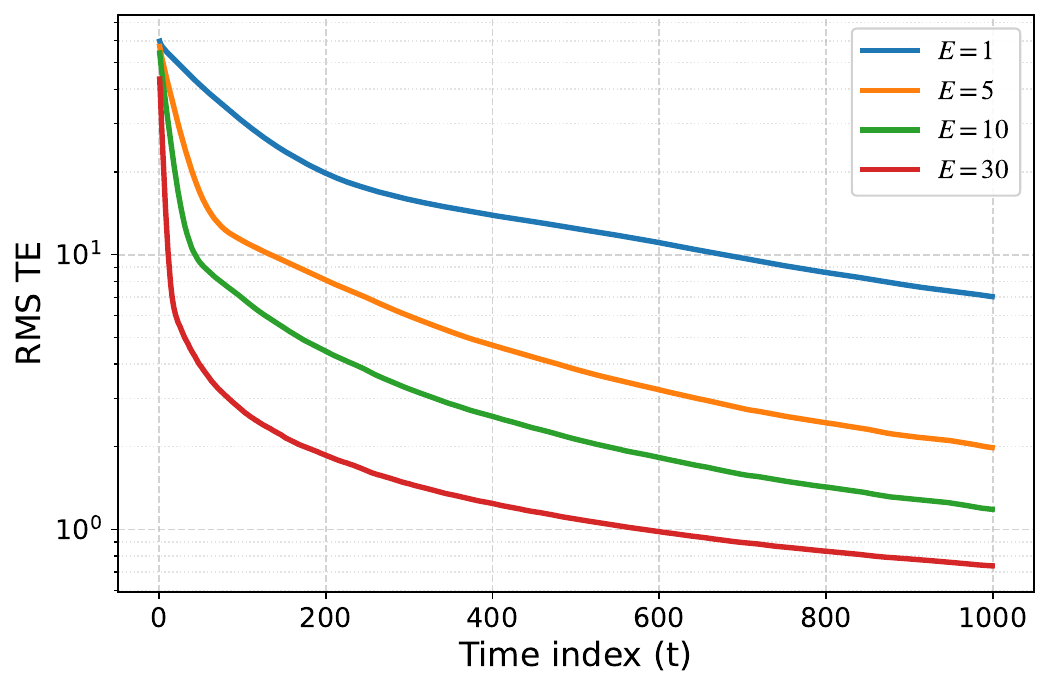}
    \caption{Diffusion with uniform temporal weights: RMS TE vs. time index for number of iterations per time-step $E$.
    }
    \label{fig:te_uniform_multiE}
\end{figure}

\begin{figure}[t]
    \centering
    \includegraphics[width=0.90\linewidth]{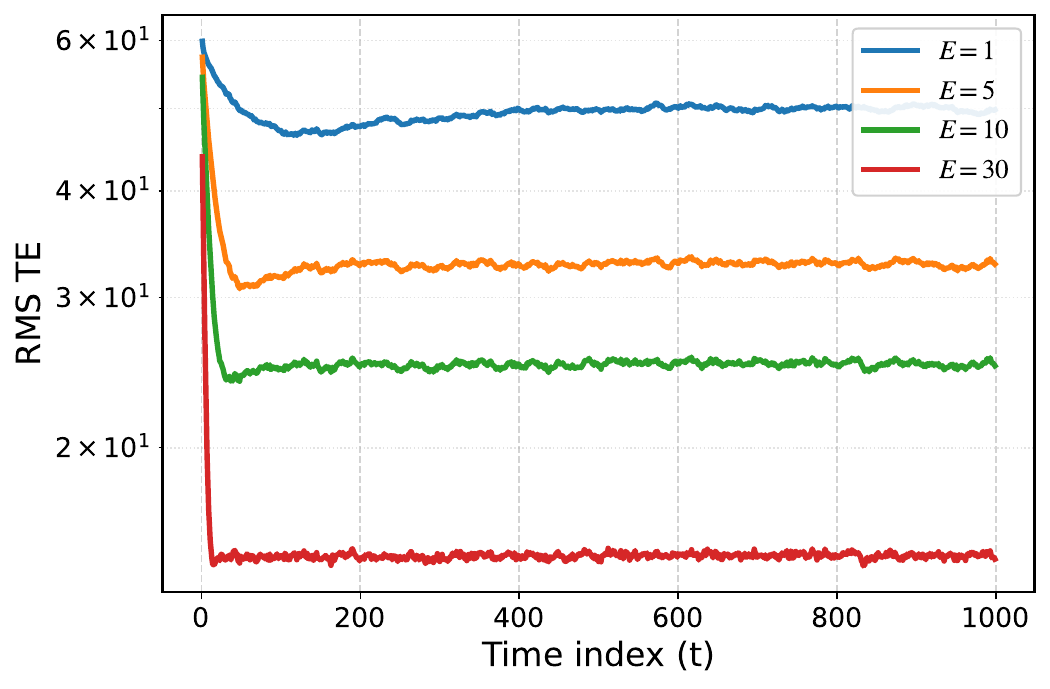}
    \caption{Diffusion with exponentially discounted temporal weights, $\gamma=0.7$: RMS TE vs. time index for number of iterations per time-step $E$.
    }
    \label{fig:te_discounted_multiE}
\end{figure}
We now present numerical experiments illustrating the tracking behavior predicted by the analysis. While the derived bounds are not expected to be tight in general, the experiments are intended to show that they are meaningful and capture the qualitative behavior of the considered decentralized system. 
At time $t$, agent $n\in\{1,\ldots,N\}$ incurs the quadratic loss
\begin{align}
\ell_{n,t}(\mathbf w)
=\frac{1}{2}
(\mathbf w-\mathbf c_{n,t})^\top
\mat A_{n,t}
(\mathbf w-\mathbf c_{n,t}),
\label{eq:num_local_quad}
\end{align}
where $\mathbf c_{n,t}\in\mathbb R^d$ models the streaming data process and $\mat A_{n,t}\in\mathbb R^{d\times d}$ is positive definite. We adopt quadratic losses because they provide a controlled and interpretable setting and admit a closed-form expression for the time-varying minimizer. This choice is appropriate because our objective is not to benchmark a particular learning architecture or dataset, but to study the qualitative implications of the theoretical analysis. We generate $\mat A_{n,t}=\mathrm{diag}(\lambda_{n,t}^1,\ldots,\lambda_{n,t}^d)$, where $\lambda_{n,t}^j\sim\mathrm{Unif}[\mu,L]$ independently across $n$, $t$, and $j$. Hence, $\ell_{n,t}$ in \eqref{eq:num_local_quad} is $\mu$-strongly convex and $L$-smooth. For this quadratic model, the temporally weighted global objective $\overline{F}_t(\mathbf w)=\sum_{i=1}^t a_i(t)F_i(\mathbf w)$  admits the closed-form minimizer $\overline{\mathbf{w}}_t^*
=
\big(\sum_{i=1}^t a_i(t)\sum_{n=1}^N \mat{A}_{n,i}\big)^{-1}
\big(\sum_{i=1}^t a_i(t)\sum_{n=1}^N \mat{A}_{n,i}\mathbf{c}_{n,i}\big)$, which allows exact evaluation of the TE. 

The temporal evolution is generated through a bounded Gaussian random walk. For each coordinate $j=1,\ldots,d$,
$
[\mathbf c_{n,t+1}]_j
=
\max\!\left\{
-C_{\max},
\min\!\left([\mathbf c_{n,t}]_j+[\mathbf z_{n,t+1}]_j,C_{\max}\right)
\right\}$,
where $\mathbf z_{n,t}\sim\mathcal N(\mathbf 0,\sigma^2\mat I)$ are independent and identically distributed across agents and time. We initialize $\mathbf c_{n,0}$ uniformly over $[-C_{\max},C_{\max}]^d$ and set $\mathbf w_{n,0}=\mathbf 0$ for all $n$.
In all experiments, we use $N=50$ agents, dimension $d=100$, step size $\eta=0.1$, $C_{\max}=10$, $\sigma^2=1$, $\mu=0.01$, and $L=0.1$. All the experiments are run for $T=1000$ time-steps.
Agents are placed uniformly at random in a disk of radius one, and an undirected random geometric graph is formed by connecting agents within distance $r_{\mathrm{th}}=0.9$, increased if needed to ensure connectivity. Given the connected graph, we construct the Metropolis-Hastings mixing matrix
\begin{align}
[\mat M]_{ij}=
\begin{cases}
\displaystyle \frac{1}{1+\max\{\deg(i),\deg(j)\}}, & (i,j)\in\mathcal E,\ i\neq j,\\[1ex]
\displaystyle 1-\sum_{j\neq i}[\mat M]_{ij}, & i=j,\\
0, & \text{otherwise},
\end{cases}
\label{eq:num_metropolis}
\end{align}
where $\deg(i)$ is the degree of agent $i$. This construction yields a symmetric doubly stochastic mixing matrix satisfying the assumptions in Section~\ref{sec:SystemModel}.
All curves report the root-mean-squared (RMS) TE $
\sqrt{\frac{1}{R}\sum_{r=1}^{R}\mathrm{TE}_r(t)^2},
$
averaged over $R=100$ Monte Carlo runs, where
$\mathrm{TE}_r(t)$ is computed according to \eqref{eq:TE_define}. The communication graph is generated once and kept fixed across all Monte Carlo runs; randomness across runs comes from the streaming data process and loss parameters.

Figures~\ref{fig:te_uniform_multiE} and~\ref{fig:te_discounted_multiE} compare different numbers of decentralized diffusion iterations per time-step $E$ for uniform and exponentially discounted temporal weights, respectively. Under uniform weighting, the influence of each new sample decreases with time, leading to a decaying fixed-point drift and a vanishing FPTE. Increasing $E$ accelerates the transient decay through the smaller contraction factor $\alpha=(1-\eta\mu)^E$. The TE therefore decreases steadily before approaching the bias-dominated regime predicted by Theorem~\ref{thm:TE_uniform}. By contrast, under discounted weighting with $\gamma=0.7$, the TE settles to a substantially larger steady-state level, consistent with Theorem~\ref{thm:TE_discounted}. This behavior arises because old samples are exponentially forgotten, so the fixed-point drift does not vanish, and the TE contains a non-zero FPTE contribution in addition to the data-heterogeneity-induced bias. Larger $E$ lowers this floor by improving tracking of the moving fixed point.

\begin{figure}[t]
    \centering
    \includegraphics[width=0.90\linewidth]{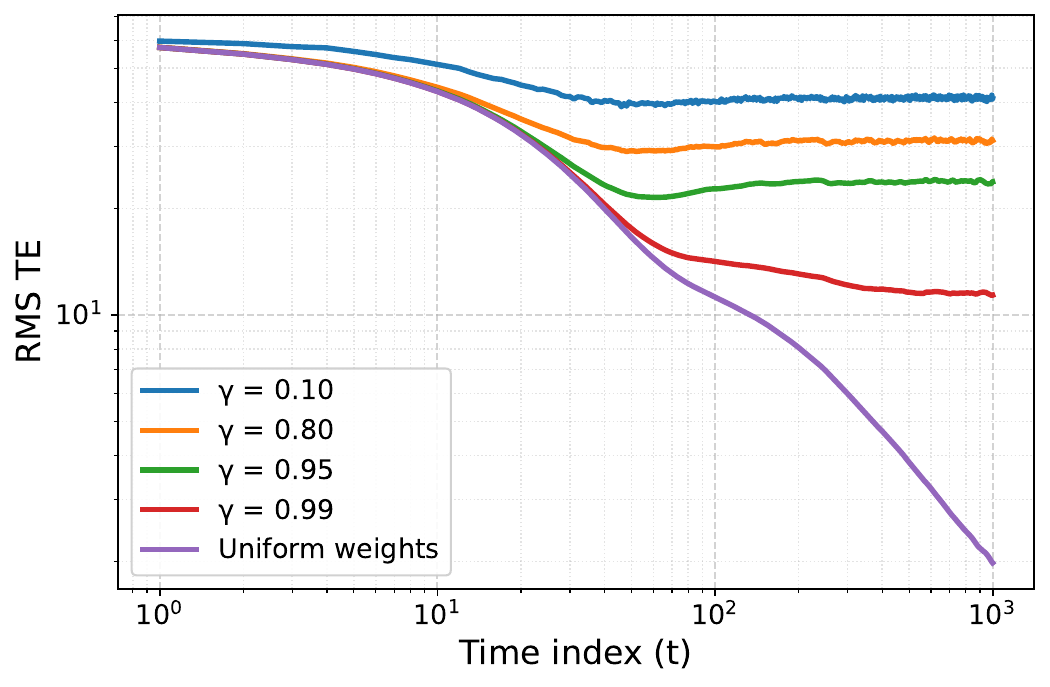}
    \caption{Diffusion with discounted weights: RMS TE vs. time index for different discount factors $\gamma$, with $E=5$. The uniform-weight curve is included as a benchmark.}
    \label{fig:te_gamma_sweep}
\end{figure}

 Figure~\ref{fig:te_gamma_sweep} studies the effect of the discount factor $\gamma$ for diffusion with fixed $E=5$. Smaller $\gamma$ places more emphasis on recent samples, making the effective objective more adaptive but also more variable; this leads to a larger asymptotic TE, consistent with Theorem~\ref{thm:TE_discounted}. As $\gamma$ increases, however, the \emph{effective memory} grows, and the TE floor decreases. In the limit $\gamma\to1$, discounted weighting approaches uniform weighting, explaining why the curve for $\gamma=0.99$ closely follows the uniform-weight curve. Moreover, the continued decay of the uniform-weight curve is consistent with
Theorem~\ref{thm:TE_uniform}: after the geometrically decaying initialization
transient, the FPTE contribution decreases as $\mathcal O(1/t)$, until the
overall TE becomes limited by the constant-step-size bias floor.

\begin{figure}[t]
    \centering
    \includegraphics[width=0.90\linewidth]{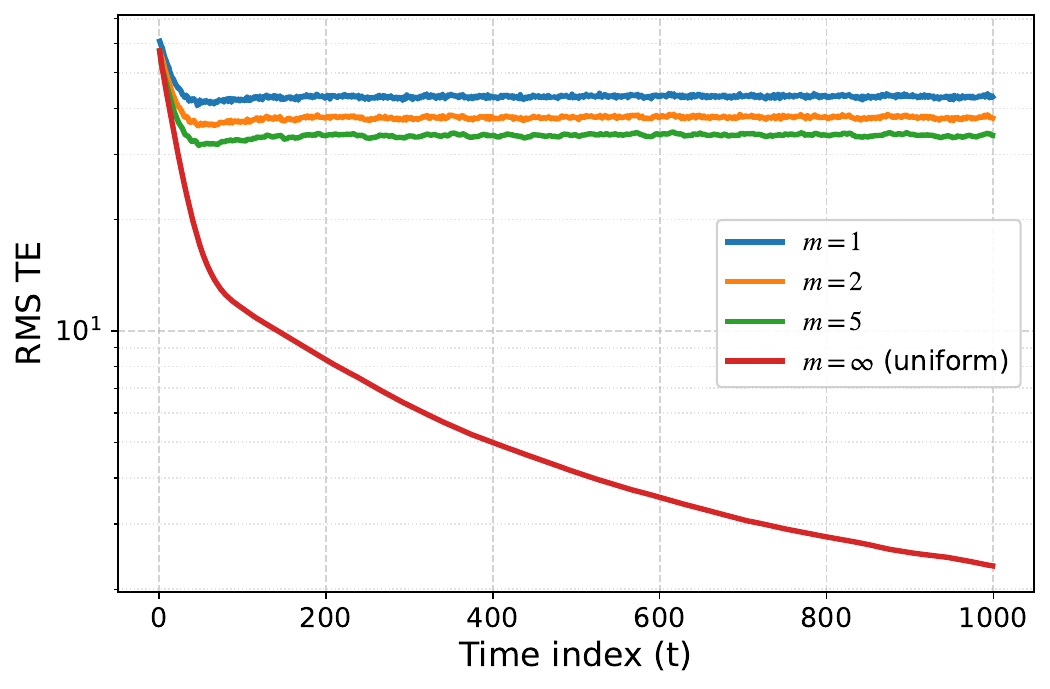}
    \caption{DGD with windowed-uniform temporal weights: RMS TE vs. time index for different window lengths $m$, with $E=5$. The case $m{=}\infty$ corresponds to standard uniform weighting.}
    \label{fig:te_window_sweep_uniform}
\end{figure}

\begin{figure}[t]
    \centering
    \includegraphics[width=0.90\linewidth]{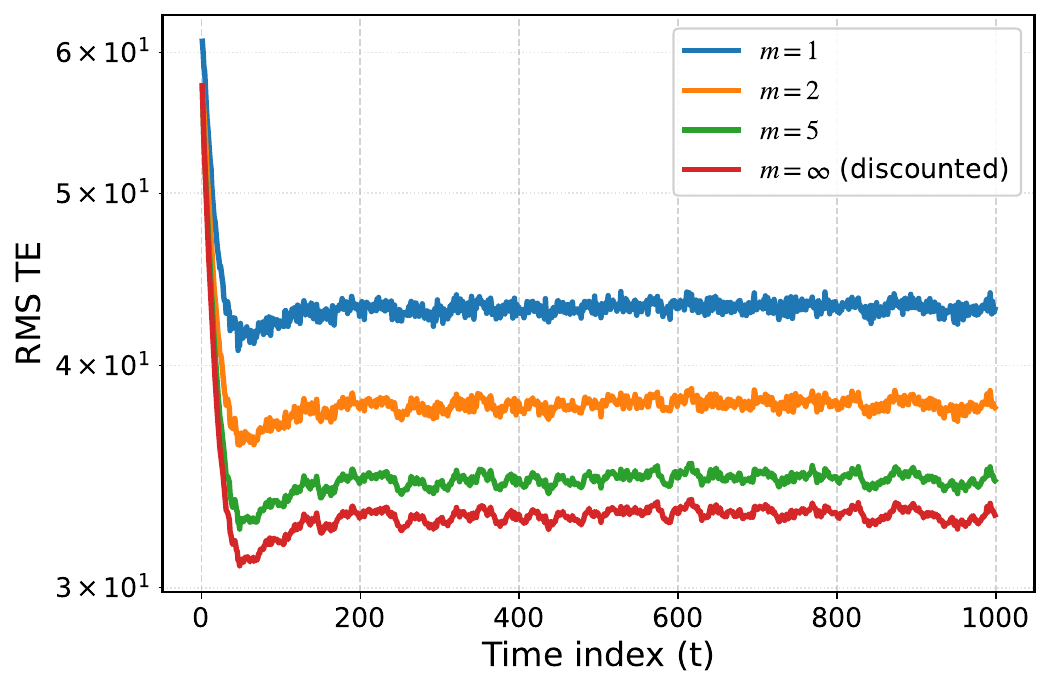}
    \caption{DGD with windowed-discounted temporal weights, $\gamma=0.7$: RMS TE vs. time index for different window lengths $m$, with $E=5$. The case $m=\infty$ corresponds to standard discounted weighting.}
    \label{fig:te_window_sweep_discounted}
    \vspace{-2mm}
\end{figure}

Figures~\ref{fig:te_window_sweep_uniform} and
\ref{fig:te_window_sweep_discounted} illustrate finite-memory windowing with
DGD. For windowed-uniform weights, increasing the window length $m$ lowers the
TE floor, in agreement with the $1/m$ dependence in
Corollary~\ref{cor:TE_window_uniform}. In contrast to the fixed-window curves,
the full-memory uniform curve, corresponding to $m=\infty$, continues to
decrease. This behavior is consistent with
Theorem~\ref{thm:TE_uniform}, which gives an $\mathcal O(1/t)$ FPTE
contribution under full-memory uniform weighting. For windowed-discounted weights,
increasing $m$ improves performance only up to the effective memory induced by
the discount factor. Notably, due to exponential forgetting, even moderate
window lengths, e.g., $m>5$, nearly match the full-memory discounted curve,
consistent with the factor $(1-\gamma)/(1-\gamma^m)$ in
\eqref{eq:ATE_window_discounted_final}.

\vspace{-3mm}
\section{Conclusion}
\label{sec:conclusion}
We investigated decentralized optimization with streaming data through a temporal-weighting formulation. For both DGD and diffusion, we developed a unified analysis that yields explicit tracking guarantees by decomposing the tracking error into a fixed-point tracking component and a network-induced bias term. We also introduced a structured class of kernel-based temporal weights, encompassing uniform and exponentially discounted rules together with memory-efficient windowed variants. Specializing the bounds, we showed that uniform weighting yields a vanishing fixed-point tracking contribution of order $\mathcal{O}(1/t)$, leaving only the decentralization-induced bias floor, whereas discounted and windowed schemes generally exhibit non-vanishing error floors governed by the discount factor and the window length, respectively. The resulting bounds explicitly show how temporal weighting, network topology, and the per-time-step iteration budget jointly determine the attainable asymptotic tracking performance. Numerical experiments corroborate the qualitative trends predicted by the theory.
\vspace{-3mm}

\appendix

\begin{lemma}[Mean Hessian Theorem \cite{MHT}]
\label{lem:MHT}
Let $g:\mathbb{R}^{Nd}\to\mathbb{R}$ be differentiable (but not necessarily twice differentiable), $\mu$-strongly convex,
and $L$-smooth. 
Then, for any $\mathbf{x},\mathbf{y}\in\mathbb{R}^{Nd}$, there exists a symmetric matrix $\mat{A}_{\mathbf{x},\mathbf{y}}\in\mathbb{R}^{Nd\times Nd}$ (dependent on $\mathbf{x}$ and $\mathbf{y}$)
such that 
\begin{align}
    &\mu \mat{I} \preceq \mat{A}_{\mathbf{x},\mathbf{y}} \preceq L \mat{I}, \label{eq:MHT-bounds}\\
    &\nabla g(\mathbf{y}) - \nabla g(\mathbf{x}) = \mat{A}_{\mathbf{x},\mathbf{y}}(\mathbf{y}-\mathbf{x}). \label{eq:MHT-relation}
\end{align}
\end{lemma}
\noindent The proof of Lemma \ref{lem:MHT} can be found in \cite{MHT}.

\begin{lemma}
\label{lemma:convergent_seq}
Let $0<\alpha<1$ and let $\{b_t\}_{t\ge0}\subset\mathbb R$ be a sequence that satisfies $b_t\to b^*$ as $t \to \infty$.
Define the sequence $\{x_t\}_{t\ge0}$ by
\begin{align}
\label{eq:linear_recursion_convergent}
x_{t+1}=\alpha x_t+ b_t,\qquad x_0\in\mathbb{R}.
\end{align}
Then, $\{x_t\}$ is also a convergent sequence which satisfies 
\[
\lim_{t\to\infty}x_t =\frac{b^*}{1-\alpha}\,.
\]
\end{lemma}

\begin{proof}
Since $b_t \to b^*$, we express $b_t$ as $b_t=b^*+e_t$ with $e_t \triangleq b_t-b^*\to0$. Subtracting the candidate limit $\frac{b^*}{1-\alpha}$ from both sides of \eqref{eq:linear_recursion_convergent}, letting $\delta_t=x_t-\frac{b^*}{1-\alpha}$ gives $\delta_{t+1}=\alpha\delta_t+e_t$. Using induction on $t$, we further obtain
\begin{align}
\delta_t
 &= \alpha^t\delta_0
 +\sum_{k=0}^{t-1}\alpha^{t-1-k}e_k.
\label{eqsx}
\end{align}
Next, we will show that for every $\varepsilon>0$ there exists a $\tau$ such that
$|\delta_t|<\varepsilon,\ \forall t\geq\tau$, which implies
 $\delta_t\to 0$, hence $x_t \to \frac{b^*}{1-\alpha}$.
Let  $\varepsilon >0$ be given. 
Since $e_t\to0$, there exists $N$ such that
$|e_t|\leq(1-\alpha)\varepsilon/2$ $\forall t\geq N$, and, since every convergent sequence is bounded, there exists $\mathcal E>0$ such that $|e_t|\leq\mathcal E ,\;\forall t$. Under such $N$, for $t\geq N$, we
rewrite $\delta_t$ as 

\begin{align*}
\delta_t
=
\alpha^t\delta_0+\underbrace{\sum_{k=0}^{N-1}\alpha^{t-1-k}e_k}_{H_t}
\;+\;\underbrace{\sum_{k=N}^{t-1}\alpha^{t-1-k}e_k}_{T_t}.
\end{align*}

Since $|e_k|\leq(1-\alpha)\varepsilon/2\,, \forall k\geq N$, we bound $|T_t|$ as
\begin{align*}
|T_t|
\leq (1-\alpha)\frac{\varepsilon}{2}\sum_{k=N}^{t-1}\alpha^{t-1-k} 
\leq (1-\alpha)\frac{\varepsilon}{2}\sum_{j=0}^{\infty}\alpha^j
=\frac{\varepsilon}{2}.
\end{align*}
Since 
$|e_t|\leq \mathcal{E},\forall t$, 
it follows that
\begin{align*}
|H_t|
\leq \sum_{k=0}^{N-1}\alpha^{t-1-k} |e_k| \leq  \mathcal{E} \sum_{k=0}^{N-1}\alpha^{t-1-k}
\leq \frac{\mathcal{E}}{1-\alpha}\,\alpha^{t-N}.
\end{align*}

Therefore, for all $t\ge N$, it holds that
\begin{align}
|\delta_t|
 &\leq \alpha^t|\delta_0|
+\frac{\ \mathcal{E}}{1-\alpha}\,\alpha^{t-N}
 +\frac{\varepsilon}{2}.
\end{align}
Since $\alpha^t\to 0$, there exists some $\tau \geq N$ such that
$\alpha^t|\delta_0|<\varepsilon/4$
and $\frac{\ \mathcal{E}}{1-\alpha}\,\alpha^{t-N}<\varepsilon/4$,
for all $t\geq \tau$,
hence $|\delta_t|<\varepsilon,\ \forall t\geq \tau$. Hence $\delta_t\to0$, which proves the lemma.
\end{proof}

\begin{proof}[Proof of Lemma~\ref{lem:generic-drift-unified}]
Fix any $t\ge 1$. From the definition of $\phi_t(\cdot)$ in \eqref{eq:phi-def-unified} and the fixed-point relation \eqref{eq:fixed-point-def-unified}, we have
\begin{align}
    (\mat{I}-\overline{\mat M})\widetilde{\bm{\mathsf w}}_t+\eta \overline{\mat Z}\nabla\overline f_t(\widetilde{\bm{\mathsf w}}_t)=\mathbf{0},
    \label{eq:fixed-point-rel-unified}
\end{align}
where $\overline{\mat{Z}}=\mat{I}_{Nd}$ for DGD and $\overline{\mat{Z}}=\overline{\mat M}=\mat M\otimes\mat I_d$ for diffusion. For notational convenience, define $\bm{\mathsf d}_t \triangleq \widetilde{\bm{\mathsf w}}_{t+1}-\widetilde{\bm{\mathsf w}}_t$.

$\bullet$ \underline{\textit{DGD ($\overline{\mat{Z}}=\mat{I}$)}:} Writing \eqref{eq:fixed-point-rel-unified} at times $t$ and $t+1$ and subtracting, we obtain:
\begin{align}
(\mat{I}-\overline{\mat{M}})\bm{\mathsf d}_t
+\eta\bigl(\nabla\overline f_{t+1}(\widetilde{\bm{\mathsf w}}_{t+1})
-\nabla\overline f_t(\widetilde{\bm{\mathsf w}}_t)\bigr)=\mathbf{0}.
\label{eq:dgd_drift_raw}
\end{align} 
Applying Lemma~\ref{lem:MHT} to $\overline f_t(\cdot)$ at the points $\widetilde{\bm{\mathsf w}}_t$ and $\widetilde{\bm{\mathsf w}}_{t+1}$, there exists a symmetric matrix $\mat{A}$ with $\mu \mat{I}\preceq \mat{A}\preceq L \mat{I}$ such that
\begin{align}
\nabla\overline f_t(\widetilde{\bm{\mathsf w}}_t)-\nabla\overline f_t(\widetilde{\bm{\mathsf w}}_{t+1})
= \mat{A}(\widetilde{\bm{\mathsf w}}_t-\widetilde{\bm{\mathsf w}}_{t+1}).
\label{eq:MHT_drift_rel_app}
\end{align}
Substituting \eqref{eq:MHT_drift_rel_app} into \eqref{eq:dgd_drift_raw} gives
\begin{align*}
    \bigl[(\mat{I}-\overline{\mat{M}})+\eta \mat{A}\bigr]\bm{\mathsf d}_t
=-\eta\bigl(\nabla\overline f_{t+1}(\widetilde{\bm{\mathsf w}}_{t+1})
-\nabla\overline f_t(\widetilde{\bm{\mathsf w}}_{t+1})\bigr).
\end{align*}
Since the eigenvalues of $\overline{\mat M}
=\mat M\otimes\mat I_d$ are those of $\mat M$, each repeated $d$
times, and $\lambda_{\max}(\mat M)=1$, it follows that
$\mat I-\overline{\mat M}\succeq0$. Moreover, since $\mat{A}\succeq \mu \mat{I}$, we have
$\lambda_{\min}((\mat{I}-\overline{\mat{M}})+\eta \mat{A})\ge \eta\mu$, and therefore it holds that
\begin{align*}\|\bm{\mathsf d}_t \|
\le \frac{\eta}{\eta\mu}\,
\bigl\|\nabla\overline f_{t+1}(\widetilde{\bm{\mathsf w}}_{t+1})
-\nabla\overline f_t(\widetilde{\bm{\mathsf w}}_{t+1})\bigr\|
\end{align*}
which yields \eqref{eq:generic-drift-unified} for any $\eta > 0$.

$\bullet$ \underline{\textit{Diffusion ($\overline{\mat{Z}}=\overline{\mat{M}}$)}:}
Specializing \eqref{eq:fixed-point-rel-unified} yields
\begin{align*}
    (\mat{I}-\overline{\mat M})\widetilde{\bm{\mathsf w}}_t+\eta \overline{\mat Z}\nabla\overline f_t(\widetilde{\bm{\mathsf w}}_t)=\mathbf{0}.
\end{align*}
Subtracting the relations at $t$ and $t+1$ and rearranging yields
\begin{align}
\mathbf{d}_t 
= \overline{\mat{M}}\bm{\mathsf d}_t 
-\eta \overline{\mat{M}}\bigl(\nabla\overline f_{t+1}(\widetilde{\bm{\mathsf w}}_{t+1})
-\nabla\overline f_t(\widetilde{\bm{\mathsf w}}_t)\bigr).
\label{eq:atc_drift_rearrange}
\end{align}
Next, applying Lemma~\ref{lem:MHT} and utilizing \eqref{eq:MHT_drift_rel_app}, \eqref{eq:atc_drift_rearrange} becomes
\begin{align*}
\bm{\mathsf d}_t
= \overline{\mat{M}}\Bigl((\mat{I}-\eta \mat{A})\bm{\mathsf d}_t
-\eta\bigl(\nabla\overline f_{t+1}(\widetilde{\bm{\mathsf w}}_{t+1})
-\nabla\overline f_t(\widetilde{\bm{\mathsf w}}_{t+1})\bigr)\Bigr).
\end{align*}
Taking norms and using  $\|\overline{\mat M}\|_2
=
\|\mat M\otimes\mat I_d\|_2
=
\|\mat M\|_2\|\mat I_d\|_2
=
1$ (since $\mat{M}$ is symmetric and doubly stochastic) gives
\begin{align*}
 \|\bm{\mathsf d}_t\|
\le \|\mat{I}-\eta \mat{A}\|_2\,\|\bm{\mathsf d}_t\|
+\eta\bigl\|\nabla\overline f_{t+1}(\widetilde{\bm{\mathsf w}}_{t+1})
-\nabla\overline f_t(\widetilde{\bm{\mathsf w}}_{t+1})\bigr\|.
\end{align*}
Next, it can be shown that under $0<\eta\le 2/(L+\mu)$, we have $\|\mat{I}-\eta \mat{A}\|_2=\max\{|1-\eta\mu|,|1-\eta L|\}\le 1-\eta\mu$,
so
\begin{align*}
\eta\mu\,\|\bm{\mathsf d}_t\|
\le
\eta\bigl\|\nabla\overline f_{t+1}(\widetilde{\bm{\mathsf w}}_{t+1})
-\nabla\overline f_t(\widetilde{\bm{\mathsf w}}_{t+1})\bigr\|,
\end{align*}
which again yields \eqref{eq:generic-drift-unified}.
\end{proof}

\begin{proof}[Proof of Lemma~\ref{lem:weight_variation_identity}]
For notational convenience, define
$
\Delta_i(t)\triangleq a_i(t+1)-a_i(t)$, for $i=1,\dots,t.
$
Then, we have
\begin{align}
\sum_{i=1}^t \Delta_i(t)
&= \sum_{i=1}^t a_i(t+1)-\sum_{i=1}^t a_i(t) \nonumber\\
&= \bigl(1-a_{t+1}(t+1)\bigr)-1 
= -a_{t+1}(t+1).
\label{eq:sum_delta_weight_var}
\end{align}
Next, we can express $ \Delta_i(t)=(\Delta_i(t))_+-(\Delta_i(t))_-$. Summing over \(i=1,\dots,t\) yields \[ \sum_{i=1}^t (\Delta_i(t))_+ - \sum_{i=1}^t (\Delta_i(t))_- = -a_{t+1}(t+1), \] where the equality follows from \eqref{eq:sum_delta_weight_var}. Hence, \[ \sum_{i=1}^t (\Delta_i(t))_- = \sum_{i=1}^t (\Delta_i(t))_+ + a_{t+1}(t+1) = I_t + a_{t+1}(t+1), \] since \(I_t=\sum_{i=1}^t (\Delta_i(t))_+\). Leveraging $|\Delta_i(t)| = (\Delta_i(t))_+ + (\Delta_i(t))_-$, we have \begin{align*} \sum_{i=1}^t |\Delta_i(t)| &= \sum_{i=1}^t \bigl((\Delta_i(t))_+ + (\Delta_i(t))_-\bigr) \\ &= I_t + \bigl(I_t + a_{t+1}(t+1)\bigr) = a_{t+1}(t+1)+2I_t, \end{align*} which proves \eqref{eq:weight_variation_identity}. \end{proof}

\begin{proof}[Proof of Lemma~\ref{lem:boundedness_compact}]
We prove the three claims in \eqref{eq:global-min-bnd}--\eqref{eq:grad-bnd-G} as follows:

$\bullet$ \underline{\textit{Bound on the global minimizer $\overline{\mathbf{w}}_t^*$}:} Since $\overline{F}_t(\cdot)$ is $\mu$-strongly convex (Assumption \ref{ass:smooth_sc}) with minimizer
$\overline{\mathbf{w}}_t^*$, we have, for any $\mathbf{w}\in\mathbb{R}^d$,
\begin{align*}
    \overline{F}_t(\mathbf{w})
    \geq \overline{F}_t(\overline{\mathbf{w}}_t^*)
      + \frac{\mu}{2}\bigl\|\mathbf{w}-\overline{\mathbf{w}}_t^*\bigr\|^2.
\end{align*}
Specializing the above inequality at $\mathbf{w}=\mathbf{0}$, we obtain:
\begin{align}
    \bigl\|\overline{\mathbf{w}}_t^*\bigr\|^2
    \leq \frac{2}{\mu}
         \bigl(\overline{F}_t(\mathbf{0}) - \overline{F}_t(\overline{\mathbf{w}}_t^*)\bigr).
         \label{eq:global-min-bnd_interm}
\end{align}
Using the definition of $\overline{F}_t(\mathbf{w})=\sum_{i=1}^t a_i(t)F_i(\mathbf{w})$ and $F_i(\mathbf{w})=\frac1N\sum_{n=1}^N \ell_{n,i}(\mathbf{w})$, we have
\begin{align*}
\overline F_t(\mathbf{0})-\overline F_t(\overline{\mathbf{w}}_t^*)
&=\sum_{i=1}^t a_i(t)\frac1N\sum_{n=1}^N\bigl(\ell_{n,i}(\mathbf{0})-\ell_{n,i}(\overline{\mathbf{w}}_t^*)\bigr) \\
&\le \sum_{i=1}^t a_i(t)\frac1N\sum_{n=1}^N\bigl(\ell_{n,i}(\mathbf{0})-\ell_{n,i}(\mathbf{w}_{n,i}^*)\bigr),
\end{align*}
where the inequality follows since 
$\mathbf{w}_{n,i}^*=\arg\min_{\mathbf{w} \in \mathbb{R}^d} \ell_{n,i}(\mathbf{w})$.
Furthermore, utilizing $L$-smoothness (Assumption \ref{ass:smooth_sc}) and optimality of $\mathbf{w}_{n,i}^*$, it follows that 
\begin{align*}
\ell_{n,i}(\mathbf{0})-\ell_{n,i}(\mathbf{w}_{n,i}^*)\le \frac{L}{2}\|\mathbf{0}-\mathbf{w}_{n,i}^*\|^2\le \frac{L}{2}C^2, 
\end{align*}
where the second inequality invokes Assumption~\ref{ass:bounded_minimizers}. Next, using $\sum_{i=1}^t a_i(t)=1$ gives
\begin{align}
    \overline F_t(\mathbf{0})-\overline F_t(\overline{\mathbf{w}}_t^*)\le \frac{L}{2}C^2.
         \label{eq:global-min-bnd_interm1}
\end{align}
Finally, substituting \eqref{eq:global-min-bnd_interm1} into \eqref{eq:global-min-bnd_interm} yields $\|\overline{\mathbf{w}}_t^*\|\le C\sqrt{L/\mu}=C\sqrt{\kappa}$, proving \eqref{eq:global-min-bnd}.

$\bullet$ \underline{\textit{Bound on the fixed point $\widetilde{\bm{\mathsf w}}_t$}:}

-- \emph{(DGD: $\mat{Z}=\mat{I}$).}
From the definition of $\phi_t(\cdot)$ in \eqref{eq:phi-def-unified} and the fixed-point relation \eqref{eq:fixed-point-def-unified}, for DGD it holds that
\begin{align}
(\mat{I}-\overline{\mat{M}})\widetilde{\bm{\mathsf w}}_t+\eta\nabla\overline f_t(\widetilde{\bm{\mathsf w}}_t)=\mathbf{0}.
\label{eq:fp_dgd}
\end{align}

Utilizing the $\mu$-strong convexity of $\overline{f}_t(\cdot)$ (Assumption \ref{ass:smooth_sc}), \eqref{eq:tilde_f_t_sc} can be specialized as
\begin{align*}
    \overline{f}_t(\mathbf{0})
    \geq   \overline{f}_t(\widetilde{\bm{\mathsf w}}_t) - \nabla \overline{f}_t(\widetilde{\bm{\mathsf w}}_t)^\top\widetilde{\bm{\mathsf w}}_t + \frac{\mu}{2}\bigl\|\widetilde{\bm{\mathsf w}}_t\bigr\|^2
\end{align*}
Next, using \eqref{eq:fp_dgd} to substitute
$\nabla\overline f_t(\widetilde{\bm{\mathsf w}}_t)=-(1/\eta)(\mat{I}-\overline{\mat{M}})\widetilde{\bm{\mathsf w}}_t$ further gives
\begin{align}
\overline f_t(\mathbf 0)
\ge \overline f_t(\widetilde{\bm{\mathsf w}}_t)
+\frac{1}{\eta}\widetilde{\bm{\mathsf w}}_t^\top(\mat{I}-\overline{\mat{M}})\widetilde{\bm{\mathsf w}}_t
+\frac{\mu}{2}\|\widetilde{\bm{\mathsf w}}_t\|^2.
\label{eq:fp_energy_dgd}
\end{align}
Since $\mat{I}-\overline{\mat{M}}\succeq 0$, the middle term is nonnegative. Hence, 
\begin{align}
\|\widetilde{\bm{\mathsf w}}_t\|^2
\le \frac{2}{\mu}\bigl(\overline f_t(\mathbf 0)-\overline f_t(\widetilde{\bm{\mathsf w}}_t)\bigr).
\label{eq:fixedpoint-bnd_interm}
\end{align}
Next, using the definition of $\overline{f}_t(\cdot)$ and $\widetilde{\bm{\mathsf w}}_t=\begin{bmatrix}
[\widetilde{\bm{\mathsf w}}_t]_1^{\top} &
        \cdots &
        [\widetilde{\bm{\mathsf w}}_t]_N^{\top}
    \end{bmatrix}^{\top}$, we have
\begin{align}
    \overline{f}_t(\mathbf{0}) - \overline{f}_t(\widetilde{\bm{\mathsf w}}_t)
       &= \sum_{i=1}^t a_i(t) 
       \sum_{n=1}^N (\ell_{n,i}(\mathbf{0}) - \ell_{n,i}([\widetilde{\bm{\mathsf w}}_t]_n))\nonumber\\
        &\overset{(a)}{\leq} \sum_{i=1}^t a_i(t) 
        \sum_{n=1}^N (\ell_{n,i}(\mathbf{0}) - \ell_{n,i}(\mathbf{w}_{n,i}^*))\nonumber\\
     &\overset{(b)}{\leq}\ \sum_{i=1}^t a_i(t)
        \sum_{n=1}^N
        \frac{L}{2}\|\mathbf{w}_{n,i}^*\|^2
     \overset{(c)}{\leq} \frac{NL}{2}C^2, \label{eq:fixedpoint-interm1}
\end{align}
where $(a)$ follows since $\mathbf{w}_{n,i}^*$ minimizes $\ell_{n,i}(\cdot)$, $(b)$ follows from $L$-smoothness (Assumption \ref{ass:smooth_sc}), and $(c)$ invokes Assumption~\ref{ass:bounded_minimizers} along with the fact that $\sum_{i=1}^t a_i(t) = 1$. Substituting \eqref{eq:fixedpoint-interm1} into \eqref{eq:fixedpoint-bnd_interm} gives
$\|\widetilde{\bm{\mathsf w}}_t\|\le C\sqrt{\frac{NL}{\mu}} = \sqrt{N}\,C_Z$, which is \eqref{eq:fixedpoint-bnd-compact} with  $C_Z=C\sqrt{\kappa}$ and holds for all $t \geq 1 $.

-- \emph{(Diffusion: $\mat{Z}=\mat{M}$).}
From \eqref{eq:phi-def-unified} and \eqref{eq:fixed-point-def-unified}, the fixed point satisfies $\widetilde{\bm{\mathsf w}}_t=\phi_t(\widetilde{\bm{\mathsf w}}_t)$ with
$\phi_t(\bm{\mathsf w})=\overline{\mat{M}}\bm{\mathsf w}-\eta \overline{\mat{M}}\nabla\overline f_t(\bm{\mathsf w})$ for diffusion. Hence, 
\begin{align*}
\|\widetilde{\bm{\mathsf w}}_t\|
&=\|\phi_t(\widetilde{\bm{\mathsf w}}_t)\|\\&\overset{(a)}{\leq}
 \|\phi_t(\widetilde{\bm{\mathsf w}}_t)-\phi_t(\mathbf 0)\|+\|\phi_t(\mathbf 0)\|
\overset{(b)}{\leq} \rho\|\widetilde{\bm{\mathsf w}}_t\|+\|\phi_t(\mathbf 0)\|,
\end{align*}
where $(a)$ follows from adding and subtracting $\phi_t(\mathbf 0)$ and using the triangle inequality, and $(b)$ invokes Lemma~\ref{lem:contraction-unified}, whereby $\phi_t(\cdot)$ is a contraction with factor $\rho\in(0,1)$. Therefore,
\begin{align}
\|\widetilde{\bm{\mathsf w}}_t\|\le \frac{1}{1-\rho}\,\|\phi_t(\mathbf 0)\|.
\label{eq:fixedpoint_bound_contraction}
\end{align}
Next, $\phi_t(\mathbf 0)=-\eta \overline{\mat{M}}\nabla\overline f_t(\mathbf 0)$, and since $\|\overline{\mat{M}}\|_2=1$ (for symmetric doubly stochastic $\overline{\mat{M}}$), we have
\begin{align*}
\|\phi_t(\mathbf 0)\|\le \eta\|\nabla\overline f_t(\mathbf 0)\|.
\end{align*}
Moreover, from \eqref{eq:fbar-grad}, the $n$-th block of $\nabla\overline f_t(\mathbf 0)$ equals
$\sum_{i=1}^t a_i(t)\nabla\ell_{n,i}(\mathbf{0})$. By $L$-smoothness (Assumption \ref{ass:smooth_sc}) and the optimality condition $\nabla\ell_{n,i}(\mathbf{w}_{n,i}^*)=\mathbf{0}$,
\begin{align*}
\|\nabla\ell_{n,i}(\mathbf{0})\|
{=}\|\nabla\ell_{n,i}(\mathbf{0}){-}\nabla\ell_{n,i}(\mathbf{w}_{n,i}^*)\|
\le L\|\mathbf{0}{-}\mathbf{w}_{n,i}^*\|
\le LC,
\end{align*}
where the last inequality utilizes Assumption~\ref{ass:bounded_minimizers}. Therefore,
\begin{align*}
\|\nabla\overline f_t(\mathbf 0)\|^2
=\sum_{n=1}^N\Big\|\sum_{i=1}^t a_i(t)\nabla\ell_{n,i}(\mathbf{0})\Big\|^2
\le \sum_{n=1}^N (LC)^2
=NL^2C^2,
\end{align*}
where the inequality follows from Jensen's inequality and $\sum_{i=1}^t a_i(t) = 1$, yielding
$\|\nabla\overline f_t(\mathbf 0)\|\le \sqrt{N}\,LC$.
Substituting into \eqref{eq:fixedpoint_bound_contraction} gives
$\|\widetilde{\bm{\mathsf w}}_t\|
\le \frac{\eta}{1-\rho}\sqrt{N}\,LC$.
Finally, under the step size range $0<\eta\le 2/(L+\mu)$, we have $\rho=1-\eta\mu$ (cf. Lemma~\ref{lem:contraction-unified}), and therefore, for all $t \geq 1 $, 
$\|\widetilde{\bm{\mathsf w}}_t\|\le C\sqrt{N}\,\kappa
$, which is \eqref{eq:fixedpoint-bnd-compact} for $\mat{Z}=\mat{M}$.

-- \emph{(Sharper diffusion bound when $\mat M\succeq0$)}. 
If, in addition, $\mat M\succeq0$, the diffusion fixed-point satisfies
\begin{align}
(\mat{I}-\overline{\mat{M}})\widetilde{\bm{\mathsf w}}_t + \eta \overline{\mat{M}}\nabla\overline f_t(\widetilde{\bm{\mathsf w}}_t)=\mathbf{0}.
\label{eq:fp_atc}
\end{align}
Define the orthogonal projectors
$
P_{\mathcal R(\overline{\mat{M}})} \triangleq \overline{\mat{M}}\,\overline{\mat{M}}^\dagger$ and $
P_{\mathcal N(\overline{\mat{M}})} \triangleq \mat{I}-\overline{\mat{M}}\,\overline{\mat{M}}^\dagger$,
where $\overline{\mat{M}}^\dagger$ is the Moore--Penrose pseudoinverse of $\overline{\mat{M}}$. 
Pre-multiplying \eqref{eq:fp_atc} by $P_{\mathcal N(\overline{\mat{M}})}$ yields 
\begin{align*}
P_{\mathcal N(\overline{\mat{M}})}\widetilde{\bm{\mathsf w}}_t =\mathbf{0},
\end{align*}
where we leverage the facts that $P_{\mathcal N(\overline{\mat{M}})}\overline{\mat{M}} = (\mat{I}-\overline{\mat{M}}\,\overline{\mat{M}}^\dagger)\overline{\mat{M}} = \overline{\mat{M}}-\overline{\mat{M}}\,\overline{\mat{M}}^\dagger \overline{\mat{M}} =\mathbf{0}$ and $P_{\mathcal N(\overline{\mat{M}})}(\mat{I}-\overline{\mat{M}})=P_{\mathcal N(\overline{\mat{M}})}-P_{\mathcal N(\overline{\mat{M}})}\overline{\mat{M}}=P_{\mathcal N(\overline{\mat{M}})}$.  Therefore, $\widetilde{\bm{\mathsf w}}_t\in\mathcal R(\overline{\mat{M}})$. Hence, there exists $\bm{\mathsf y}_t \in \mathbb{R}^{Nd}$ such that
\begin{align}
\widetilde{\bm{\mathsf w}}_t = \overline{\mat{M}} \bm{\mathsf y}_t. 
\label{eq:ATC_refined_fp_range_M}
\end{align}
Next, pre-multiply \eqref{eq:fp_atc} by $\bm{\mathsf y}_t^\top$ to obtain
\begin{align*}
\bm{\mathsf y}_t^\top(\mat{I}-\overline{\mat{M}})\widetilde{\bm{\mathsf w}}_t + \eta\,\bm{\mathsf y}_t^\top \overline{\mat{M}}\nabla\overline f_t(\widetilde{\bm{\mathsf w}}_t)=0.
\end{align*}
Since $\overline{\mat{M}}$ is symmetric, $\bm{\mathsf y}_t^\top \overline{\mat{M}} = (\overline{\mat{M}}\bm{\mathsf y}_t)^\top = \widetilde{\bm{\mathsf w}}_t^\top$, we further obtain
\begin{align}
\nabla\overline f_t(\widetilde{\bm{\mathsf w}}_t)^\top \widetilde{\bm{\mathsf w}}_t
= -\frac{1}{\eta}\,\bm{\mathsf y}_t^\top(\mat{I}-\overline{\mat{M}})\widetilde{\bm{\mathsf w}}_t
= -\frac{1}{\eta}\,\bm{\mathsf y}_t^\top(\mat{I}-\overline{\mat{M}})\overline{\mat{M}} \bm{\mathsf y}_t,
\end{align}
where the last equality uses \eqref{eq:ATC_refined_fp_range_M}. 
Since $\mat M\succeq0$, 
$\overline{\mat M}$ is also positive
semidefinite and its eigenvalues are those of $\mat M$. Hence, since $\lambda(\overline{\mat{M}})\subset[0,1]$, the matrix $(\mat{I}-\overline{\mat{M}})\overline{\mat{M}}$ is positive semidefinite (its eigenvalues are $\lambda(1-\lambda)\ge 0$), and therefore,
\begin{align}
-\nabla\overline f_t(\widetilde{\bm{\mathsf w}}_t)^\top \widetilde{\bm{\mathsf w}}_t
= \frac{1}{\eta}\,\bm{\mathsf y}_t^\top (\mat{I}-\overline{\mat{M}})\overline{\mat{M}}\bm{\mathsf y}_t \ge 0.
\label{eq:innerprod_atc}
\end{align}
Next, utilizing the $\mu$-strong convexity of $\overline{f}_t(\cdot)$ (Assumption \ref{ass:smooth_sc}), \eqref{eq:tilde_f_t_sc} can be specialized as
\begin{align*}
\overline f_t(\mathbf 0)
{\ge} \overline f_t(\widetilde{\bm{\mathsf w}}_t)
-\nabla\overline f_t(\widetilde{\bm{\mathsf w}}_t)^\top\widetilde{\bm{\mathsf w}}_t
{+}\frac{\mu}{2}\|\widetilde{\bm{\mathsf w}}_t\|^2
{\ge} \overline f_t(\widetilde{\bm{\mathsf w}}_t)+\frac{\mu}{2}\|\widetilde{\bm{\mathsf w}}_t\|^2,
\end{align*}
where the last inequality uses \eqref{eq:innerprod_atc}. Rearranging yields
\begin{align*}
\|\widetilde{\bm{\mathsf w}}_t\|^2
\le \frac{2}{\mu}\bigl(\overline f_t(\mathbf 0)-\overline f_t(\widetilde{\bm{\mathsf w}}_t)\bigr),
\end{align*}
which is identical to \eqref{eq:fixedpoint-bnd_interm}. Thus, following the same steps as in \eqref{eq:fixedpoint-bnd_interm}--\eqref{eq:fixedpoint-interm1} yields the same bound as that obtained for DGD:  
$\|\widetilde{\bm{\mathsf w}}_t\|\le C\sqrt{\frac{NL}{\mu}} = \sqrt{N}\,C_Z$, which is \eqref{eq:fixedpoint-bnd-compact} with  $C_Z=C\sqrt{\kappa}$.



$\bullet$ \underline{\textit{Gradient Bounds $\|\nabla\overline f_t(\widetilde{\bm{\mathsf w}}_i)\|$ and $
\|\nabla f_t(\widetilde{\bm{\mathsf w}}_i)\|$}:} We begin by recalling the $\mu$-strong convexity of $\overline{f}_t(\cdot)$ under Assumption~\ref{ass:smooth_sc}. Let $\overline{\bm{\mathsf m}}_t^*=\arg\min_{\bm{\mathsf w}\in\mathbb{R}^{Nd}}\overline f_t(\bm{\mathsf w})$ denote the minimizer of $\overline f_t(\cdot)$.
Next, since $\overline f_t(\bm{\mathsf w})=\sum_{n=1}^N\sum_{i=1}^t a_i(t)\ell_{n,i}([\bm{\mathsf w}]_n)$ is block-separable, each $d$-dimensional block $[\overline{\bm{\mathsf m}}_t^*]_n$ minimizes the function $\sum_{i=1}^t a_i(t)\ell_{n,i}(\mathbf{w})$ with $\mathbf{w}\in\mathbb{R}^d$. As a result, by using the same argument used to bound $\overline{\mathbf{w}}_t^*$ in \eqref{eq:global-min-bnd_interm}--\eqref{eq:global-min-bnd_interm1}, it is easy to verify that $\bigl\|[\overline{\bm{\mathsf m}}_t^*]_n\bigr\|
\le C\sqrt{\kappa}$, and hence 
\begin{align}
\|\overline{\bm{\mathsf m}}_t^*\|\le \sqrt{N}\,C_Z.
\label{eq:gradient-bnd_interm}
\end{align} 
Next, $L$-smoothness (Assumption \ref{ass:smooth_sc}) of $\overline f_t(\cdot)$ and the optimality condition $\nabla\overline f_t(\overline{\bm{\mathsf m}}_t^*)=\mathbf 0$ yields, for any $i,t\ge 1$,
\begin{align*}
\|\nabla\overline f_t(\widetilde{\bm{\mathsf w}}_i)\|
&= \|\nabla\overline f_t(\widetilde{\bm{\mathsf w}}_i)-\nabla\overline f_t(\overline{\bm{\mathsf m}}_t^*)\|\\
&{\leq} L\|\widetilde{\bm{\mathsf w}}_i-\overline{\bm{\mathsf m}}_t^*\|
\overset{(a)}{{\leq}}  L\bigl(\|\widetilde{\bm{\mathsf w}}_i\|+\|\overline{\bm{\mathsf m}}_t^*\|\bigr)
\overset{(b)}{\leq}2L\sqrt{N}\,C_Z,
\end{align*}
where $(a)$ invokes the triangle inequality and $(b)$ uses \eqref{eq:gradient-bnd_interm} and $\|\widetilde{\bm{\mathsf w}}_t\|\le \sqrt{N}\,C_Z$. This proves the first inequality in \eqref{eq:grad-bnd-G} with $G_Z=2L\sqrt{N}C_Z$.

Let $\bm{\mathsf m}_t^*=\arg\min_{\bm{\mathsf w}} f_t(\bm{\mathsf w})$. Since $f_t(\bm{\mathsf w})=\sum_{n=1}^N \ell_{n,t}([\bm{\mathsf w}]_n)$ is block-wise separable, it follows that $\bm{\mathsf m}_t^*=\begin{bmatrix}
(\mathbf w_{1,t}^*)^\top 
\cdots 
(\mathbf w_{N,t}^*)^\top
\end{bmatrix}^{\top}$, hence $\|\bm{\mathsf m}_t^*\|\le \sqrt{N}C\le \sqrt{N}C_Z$. Using $L$-smoothness (Assumption \ref{ass:smooth_sc}) of $f_t(\cdot)$ and $\nabla f_t(\bm{\mathsf m}_t^*)=\mathbf 0$ yields
\begin{align*}
\|\nabla f_t(\widetilde{\bm{\mathsf w}}_i)\|
\le L\|\widetilde{\bm{\mathsf w}}_i-\bm{\mathsf m}_t^*\|
\le L(\|\widetilde{\bm{\mathsf w}}_i\|+\|\bm{\mathsf m}_t^*\|)
\le 2L\sqrt{N}C_Z,
\end{align*}
for any $i,t\ge 1$, which proves the second inequality in \eqref{eq:grad-bnd-G}.
Finally, since $\overline{\bm{\mathsf w}}_t^*=\bm 1_N\otimes\overline{\mathbf w}_t^*$, \eqref{eq:global-min-bnd} implies
$\|\overline{\bm{\mathsf w}}_t^*\|= \|\bm 1_N\otimes\overline{\mathbf w}_t^*\| =
\sqrt N\,
\|\overline{\mathbf w}_t^*\|\le C\sqrt{N \kappa}\,\le \sqrt{N}\,C_Z$. Using again the $L$-smoothness of $\overline f_t(\cdot)$ and the optimality condition
$\nabla \overline f_t(\overline{\bm{\mathsf m}}_t^*)=\mathbf 0$, for all $t \geq 1$, we have
\begin{align*}
\|\nabla \overline f_t(\overline{\bm{\mathsf w}}_t^*)\|
&{=} \|\nabla \overline f_t(\overline{\bm{\mathsf w}}_t^*)-\nabla \overline f_t(\overline{\bm{\mathsf m}}_t^*)\|\\
&{\le} L\|\overline{\bm{\mathsf w}}_t^*-\overline{\bm{\mathsf m}}_t^*\|
{\le} L\bigl(\|\overline{\bm{\mathsf w}}_t^*\|+\|\overline{\bm{\mathsf m}}_t^*\|\bigr)
{\le} 2L\sqrt{N}\,C_Z.
\end{align*}
Utilizing the definition of $G_Z$ yields the desired bound.
\end{proof}

\begin{proof}[Proof of Theorem~\ref{thm:TE_uniform}]
We first show that for all $t\ge t_0$ it holds that 
\begin{align}
S(t)\le \frac{A}{t}.
\label{eq:S_uniform_bound}
\end{align}
Note that we can express $S(t+1)$ as
    \begin{align}
        S(t+1) &=\sum_{i=1}^{t}  \frac{\alpha^{t+1-i}}{i+1}= \frac{\alpha}{t+1} + \alpha S(t).\label{Uniform_Sum_recursive}
    \end{align}
By the definition of $A$, we have $A \geq t_0S(t_0)$, hence it directly holds that $S(t_0) \leq \frac{A}{t_0}$. Next, as induction hypothesis assume that $S(t) \leq \frac{A}{t}$ for some $t\geq t_0 = \lceil \frac{2\alpha}{1-\alpha}\rceil$. 
    Then, from \eqref{Uniform_Sum_recursive}, $S(t+1)$ can be upper bounded as $S(t+1) \leq \frac{\alpha}{t+1} + \alpha \frac{A}{t}$. To complete the argument, it is sufficient to show that the right-hand side above is at most $\frac{A}{t+1}$, which is equivalent to showing that 
    \begin{align*}
        A\big(1 - \alpha -\frac{\alpha}{t}\big)\geq\alpha.
    \end{align*}
    Since $t\geq t_0\geq 2\alpha/(1-\alpha)$, it follows that 
    $1-\alpha-\alpha/t\geq(1-\alpha)/2$. Thus, the desired inequality holds whenever $A\ge 2\alpha/(1-\alpha)$, which is true by definition of $A$. 
    Therefore, it follows by induction that \eqref{eq:S_uniform_bound} holds for all $t\ge t_0$. Leveraging this bound in \eqref{eq:TE_uniform} yields the desired TE bound \eqref{eq:TE_uniform_final}. Finally, taking $\limsup_{t\to\infty}$ in \eqref{eq:TE_uniform_final} gives the stated ATE bound,  completing the proof.
\end{proof}
\begin{proof}[Proof of Theorem~\ref{thm:TE_discounted}]
We first show that, for all $t\ge t_0$,
\begin{align}
S_\gamma(t)
\le
A_\gamma\frac{1-\gamma}{1-\gamma^t}.
\label{eq:S_discounted_bound}
\end{align}
Note that $S_\gamma(t+1) = \sum_{i=1}^{t}  \frac{(1-\gamma)\alpha^{t+1-i}}{1 - \gamma^{i+1}} $ can be written recursively as 
     \begin{align}
        S_\gamma(t+1) &=   \frac{(1-\gamma)\alpha}{1-\gamma^{t+1}} + \alpha S_\gamma(t).\label{Discount_Sum_recursive}
    \end{align}
    
By definition of $A_\gamma$, we have $A_\gamma \geq \frac{(1-\gamma^{t_0})S_\gamma(t_0)}{1-\gamma}$, thus it holds that $S_\gamma(t_0) \leq \frac{A_\gamma(1-\gamma)}{1-\gamma^{t_0}}$. Next, we assume as induction hypothesis that $ S_\gamma(t) \leq \frac{A_\gamma(1-\gamma)}{1-\gamma^t}$ for some $t\geq t_0 =\lceil\ln(\frac{1-\alpha}{1+\alpha-2\gamma\alpha})/\ln(\gamma)\rceil$. Then, from \eqref{Discount_Sum_recursive}, $S_\gamma(t+1)$ can be upper bounded as $S_\gamma(t+1) \leq  \frac{(1-\gamma)\alpha}{1-\gamma^{t+1}} + \alpha  \frac{A_\gamma(1-\gamma)}{1-\gamma^t}$. To complete the induction, it suffices to show that the right-hand side above is at most $\frac{A_\gamma(1-\gamma)}{1-\gamma^{t+1}}$. After rearranging and simplifying, this condition is equivalent to
\begin{align}
    A_\gamma\Big[
1-\alpha-\alpha\gamma^t\frac{
1-\gamma}{1-\gamma^t}
\Big]
\geq \alpha.
\label{eq:discount_induction_condition}
\end{align}
Since $t\ge t_0$ and $\gamma\in(0,1)$, the definition of $t_0$ implies
    \begin{align*}
    \gamma^t\leq \gamma^{t_0}
    \leq
    \gamma^{\ln(\frac{1-\alpha}{1+\alpha-2\gamma\alpha})/\ln(\gamma)}
    =
\frac{1-\alpha}{1+\alpha-2\gamma\alpha}.
    \end{align*}
    Therefore, we can lower bound the left-hand side of \eqref{eq:discount_induction_condition} as
    $$
    A_\gamma\Big[
1-\alpha-\alpha\gamma^t\frac{
1-\gamma}{1-\gamma^t}
\Big]
$$$$
\geq
A_\gamma\Big[
1-\alpha-\alpha
\frac{1-\alpha}{1+\alpha-2\gamma\alpha}
\frac{
1-\gamma}{1-\frac{1-\alpha}{1+\alpha-2\gamma\alpha}}
\Big]
=
A_\gamma\frac{1-\alpha}{2},
    $$
    where the inequality is due to the bound on $\gamma^t$.
    Finally, by the definition of 
    $A_\gamma$, we have
    $A_\gamma(\frac{1-\alpha}{2})\geq \alpha$, 
    and thus \eqref{eq:discount_induction_condition} holds.
We have thus proved that $S_\gamma(t+1)\leq \frac{A_\gamma(1-\gamma)}{1-\gamma^{t+1}}$. Hence, by induction, \eqref{eq:S_discounted_bound} holds for all $t\ge t_0$. Substituting \eqref{eq:S_discounted_bound} into \eqref{eq:TE_discounted_pre} yields \eqref{eq:TE_discounted_final}.

Next, from \eqref{Discount_Sum_recursive}, the sequence $\{S_\gamma(t)\}$ satisfies
\begin{align*}
S_\gamma(t+1)=\alpha S_\gamma(t)+b_t,
\end{align*}
with $b_t\triangleq \frac{(1-\gamma)\alpha}{1-\gamma^{t+1}}$. Since $b_t\to (1-\gamma)\alpha$ and $\alpha\in(0,1)$, Lemma~\ref{lemma:convergent_seq} immediately  yields
\begin{align*}
 \lim_{t\to\infty}S_\gamma(t)
=
\frac{(1-\gamma)\alpha}{1-\alpha}.
\end{align*}
Applying $\limsup_{t \to \infty}$ to both sides of \eqref{eq:TE_discounted_pre}, and using the limit of $S_\gamma(t)$ gives \eqref{eq:ATE_discounted_final}, which completes the proof.
\end{proof}
\vspace{-5mm}
\bibliographystyle{IEEEtran}
\bibliography{Refs}

\end{document}